\documentclass{vldb}

\usepackage[show]{chato-notes}
\usepackage{rotating}
\usepackage{graphicx}
\usepackage{color}
\usepackage{cite}
\usepackage{url}

\usepackage{amsmath,amssymb, amsthm}
\usepackage{dsfont}
\usepackage{epsfig}
\usepackage{framed}
\usepackage[ruled,lined,linesnumbered,noend]{algorithm2e}
\usepackage{multirow}
\usepackage{capt-of}
\usepackage{textcomp}
\usepackage{tabularx}
\usepackage{bbm}
\usepackage{times}
\usepackage{hyperref}

\hypersetup{
    bookmarks=true,         
    unicode=false,          
    pdftoolbar=true,        
    pdfmenubar=true,        
    pdffitwindow=false,     
    pdfstartview={FitH},    
    pdftitle={My title},    
    pdfauthor={Author},     
    pdfsubject={Subject},   
    pdfcreator={Creator},   
    pdfproducer={Producer}, 
    pdfnewwindow=true,      
    colorlinks=false,       
    linkcolor=black,          
    citecolor=black,        
    filecolor=black,      
    urlcolor=black           
}

\newcommand{\RNum}[1]{\uppercase\expandafter{\romannumeral #1\relax}}

\newcommand{\vecalloc}{\vec{S_g}}
\newcommand{\vecall}{\vec{S}}
\DeclareMathOperator*{\argmin}{\arg\!\min~}
\DeclareMathOperator*{\argmax}{\arg\!\max~}

\SetKwRepeat{Do}{do}{while}

\newcommand{\CARM}{\textsc{CA-Greedy}\xspace}
\newcommand{\CSRM}{\textsc{CS-Greedy}\xspace}
\newcommand{\RM}{\textsc{RM}\xspace}
\newcommand{\RMLong}{\textsc{Revenue-Maximization}\xspace}
\newcommand{\cpe}[1]{cpe({#1})}

\newcommand{\fastca}{\textsc{TI-CARM}\xspace}
\newcommand{\fastcs}{\textsc{TI-CSRM}\xspace}

\newcommand{\spara}[1]{\smallskip\noindent{\bf #1}}

\newtheorem{definition}{Definition}

\newtheorem{theorem}{Theorem}

\newtheorem{lemma}{Lemma}
\newtheorem{problem}{Problem}

\newtheorem{observation}{Observation}

\newcommand{\Card}[1]{\left\vert{#1}\right\vert}

\newcommand{\LL}[1]{\textcolor{black}{#1}}
\newcommand{\CA}[1]{\textcolor{black}{#1}}

\newcommand{\AR}[1]{\textcolor{black}{#1}}

\newcommand{\ravi}[1]{\textcolor{black}{#1}}
\newcommand{\CN}[1]{\textcolor{black}{#1}}
\newcommand{\CC}[1]{\textcolor{black}{#1}}
\newcommand{\sstar}{S^*} 
\newcommand{\ssample}{S^+} 
\newcommand{\stilde}{\tilde{S}}
\newcommand{\vecstilde}{\vec{\tilde{S}}}
\newcommand{\vecssample}{\vec{S}^+}
\newcommand{\vecsstar}{\vec{S}^*}
\newcommand{\isfrac}[1]{F_{\RR_i}(#1)}
\newcommand{\spi}{\tilde{\pi}} 
\newcommand{\srho}{\tilde{\rho}} 
\newcommand{\rbudget}[1]{\tilde{B}_{#1}} 
\newcommand{\etb}[2]{\eta_{#1,#2}} 
\newcommand{\sib}[1]{\bar{s}_{#1}} 
\newcommand{\prob}[1]{\ensuremath{\mathrm{Pr}}\left[#1\right]}

\newcommand{\NPhard}{NP-hard\xspace}

\newcommand{\SPhard}{\#P-hard\xspace}

\newcommand{\E}{\mathbb{E}}
\newcommand{\RR}{\mathbf{R}}

\newcommand{\eat}[1]{}

\newcommand{\squishlist}{
 \begin{list}{$\bullet$}
  {  \setlength{\itemsep}{0pt}
     \setlength{\parsep}{3pt}
     \setlength{\topsep}{3pt}
     \setlength{\partopsep}{0pt}
     \setlength{\leftmargin}{2em}
     \setlength{\labelwidth}{1.5em}
     \setlength{\labelsep}{0.5em}
} }
\newcommand{\squishlisttight}{
 \begin{list}{$\bullet$}
  { \setlength{\itemsep}{0pt}
    \setlength{\parsep}{0pt}
    \setlength{\topsep}{0pt}
    \setlength{\partopsep}{0pt}
    \setlength{\leftmargin}{2em}
    \setlength{\labelwidth}{1.5em}
    \setlength{\labelsep}{0.5em}
} }

\newcommand{\squishdesc}{
 \begin{list}{}
  {  \setlength{\itemsep}{0pt}
     \setlength{\parsep}{3pt}
     \setlength{\topsep}{3pt}
     \setlength{\partopsep}{0pt}
     \setlength{\leftmargin}{1em}
     \setlength{\labelwidth}{1.5em}
     \setlength{\labelsep}{0.5em}
} }

\newcommand{\squishend}{
  \end{list}
}

\newcommand{\flix}{\textsc{Flixster}\xspace}

\newcommand{\livej}{\textsc{LiveJournal}\xspace}

\newcommand{\epi}{\textsc{Epinions}\xspace}
\newcommand{\dblp}{\textsc{DBLP}\xspace}

\newcommand{\ra}{\rightarrow}

\newcommand{\calX}{\mbox{$\mathcal{X}$}}
\newcommand{\calE}{\mbox{$\mathcal{E}$}}

\newcommand{\calC}{\mbox{$\mathcal{C}$}}
\newcommand{\calI}{\mbox{$\mathcal{I}$}}
\newcommand{\calU}{\mbox{$\mathcal{U}$}}

\newcommand{\posreals}{\mathbb{R}_{\ge 0}}

\begin{document}
\title{Revenue Maximization in Incentivized Social Advertising}
\author{\begin{tabular}{cccc}
Cigdem Aslay & \hspace{2mm}  Francesco Bonchi &  \hspace{1mm} Laks V.S. Lakshmanan &  Wei Lu\\
\affaddr{ISI Foundation} & \hspace{2mm}  \affaddr{ISI Foundation} &  \hspace{1mm} \affaddr{Univ. of British Columbia} &  \affaddr{LinkedIn Corp.}\\
\affaddr{Turin, Italy} & \hspace{2mm} \affaddr{Turin, Italy}  &  \hspace{1mm} \affaddr{Vancouver, Canada} & \affaddr{Sunnyvale, CA, USA}\\
{\sf cigdem.aslay@isi.it} & \hspace{2mm}  {\sf francesco.bonchi@isi.it} &  \hspace{1mm}{\sf laks@cs.ubc.ca} & {\sf wlu@linkedin.com}
\end{tabular}
}

\maketitle
\sloppy

%


\begin{abstract}
\enlargethispage{4\baselineskip}
Incentivized social advertising, an emerging marketing model, provides monetization opportunities not only to the owners of the social networking  platforms but also to their influential users by offering a ``cut'' on the advertising revenue. We consider a social network (the host) that sells ad-engagements to advertisers by inserting their ads, in the form of promoted posts, into the feeds of carefully selected ``initial endorsers" or seed users: these users receive monetary incentives in exchange for their endorsements. The endorsements help propagate the ads to the feeds of their followers. Whenever any user of the platform engages with an ad, the host is paid some fixed amount by the advertiser, and the ad further propagates to the feed of her followers, potentially recursively. In this context, the problem for the host is is to allocate ads to influential users, taking into account the propensity of ads for viral propagation, and carefully apportioning the monetary budget of each of the advertisers between incentives to influential users and ad-engagement costs, with the rational goal of maximizing its own revenue.
In particular, we consider a monetary incentive for the influential users, which is proportional to their influence potential.

We show that, taking all important factors into account, the problem of revenue maximization in incentivized social advertising corresponds to the problem of monotone submodular function maximization, subject to a partition matroid constraint on the ads-to-seeds allocation, and submodular knapsack constraints on the advertisers' budgets. We show that this problem is NP-hard and devise two greedy algorithms with provable approximation guarantees, which differ in their sensitivity to seed user incentive costs.

Our approximation algorithms require repeatedly estimating the expected marginal gain in revenue as well as in advertiser payment. By exploiting a connection to the recent advances made in scalable estimation of expected influence spread, we devise efficient and scalable versions of our two greedy algorithms. An extensive experimental assessment confirms the high quality of our proposal.

\end{abstract}

\thispagestyle{empty}

\section{Introduction}
\label{sec:intro}

The rise of online advertising platforms has generated new opportunities for advertisers in terms of personalizing and targeting their marketing messages. When users access a platform, they leave a trail of information that can be correlated with their consumption tastes, enabling better targeting options for advertisers. Social networking platforms particularly can gather large amounts of users' shared posts that stretches beyond general demographic and geographic data. This offers more advanced interest, behavioral, and connection-based targeting options, enabling a level of personalization that is not achievable by other online advertising channels. Hence, advertising on social networking platforms has been one of the fastest growing sectors in the online advertising landscape: a market that did not exist until Facebook launched its first advertising service in May $2005$, is projected to generate $\$11$ billion revenue by $2017$, almost doubling the $2013$ revenue\footnote{\scriptsize \url{http://www.unified.com/historyofsocialadvertising/}}.

\spara{Social advertising.} Social advertising models are typically employed by platforms such as Twitter, Tumblr, and Facebook through the implementation of \emph{promoted posts} that are shown in the ``news feed" of their users.\footnote{\scriptsize According to a recent report, Facebook's news feed ads have $21$ times higher click-through rate than standard web retargeting ads and  $49$ times the click-through rate of Facebook's right-hand side display ads: see \url{https://blog.adroll.com/trends/facebook-exchange-news-feed-numbers}.} A promoted post can be a video, an image, or simply a textual post containing an advertising message.
Social advertising models of this type are usually associated with a \emph{cost per engagement} (CPE) pricing scheme: the advertiser does not pay for the ad impressions, but pays the platform owner (hereafter referred to as the \emph{host}) only when a user actively engages with the ad. The \emph{engagement} can be in the form of a social action such as \emph{``like"}, \emph{``share''}, or \emph{``comment''}: in this paper we blur the distinction between these different types of actions, and generically refer to them all as  \emph{engagements} or \emph{clicks} interchangeably.

Similar to organic (i.e., non-promoted) posts, promoted posts can propagate from user to user in the network\footnote{\scriptsize Tumblr's CEO D. Karp reported (CES 2014) that a normal post is reposted on average 14 times, while promoted posts are on average reposted more than 10\,000 times: \url{http://yhoo.it/1vFfIAc}.}, potentially triggering a viral contagion: whenever a user $u$ engages with an ad $i$, the host is paid some fixed amount by the advertiser (the CPE). Furthermore, $u$'s engagement with $i$ appears in the feed of $u$'s followers, who are then exposed to ad $i$ and could in turn be influenced to engage with $i$, producing further revenue for the host~\cite{bakshy12,tucker12}.



\enlargethispage{2\baselineskip}
\spara{Incentivized social advertising.} In this paper, we study the novel model of \emph{incentivized social advertising}. Under this model, users selected by the host as  \emph{seeds} for the campaign on a specific ad $i$, can take a ``cut'' on the social advertising revenue. These users are typically selected because they are influential or authoritative on the specific topic, brand, or market of $i$.

A recent report\footnote{\scriptsize \url{http://www.theverge.com/2016/4/19/11455840/facebook-tip-jar-partner-program-monetization}} indicates that Facebook is experimenting with the idea of incentivizing users. YouTube launched a revenue-sharing program for prominent users in 2007. Twitch, the streaming platform of choice for gamers, lets partners make money through revenue sharing, subscriptions, and merchandise sales. YouNow, a streaming platform popular among younger users, earns money by taking a cut of the tips and digital gifts that fans give its stars. On platforms without partner deals, including Twitter and Snapchat, celebrity users often strike sponsored deals to include brands in their posts, which suggests potential monetization opportunities for Twitter and Snapchat\footnote{\scriptsize \url{http://www.wsj.com/articles/more-marketers-offer-incentives-for-watching-ads-1451991600}}.


In this work, we consider incentives that are
determined by the topical influence of the seed users for the specific ad. More concretely, given an ad $i$, the financial incentive that a seed user $u$ would get for engaging with $i$ is a function of the social influence that $u$ has exhibited in the past in the topic of $i$. For instance, a user who  often produces relevant content about long-distance running, capturing the attention of a relatively large audience, might be a good seed for endorsing a new model of running shoes. In this case, her past demonstrated influence on this very topic would be taken into consideration when defining the lumpsum amount for her engagement with the new model of running shoes. The same user could be considered as a seed for a new model of tennis shoes, but in that case the incentive might be lower, due to her lower past influence demonstrated. To summarize, incentives are paid by the host to users selected as seeds. These incentives count as seeding costs and depend on the topic of the ad and the user's past demonstrated influence in the topic.

The incentive model above has several advantages. First, it captures in a uniform framework both the ``celebrity-influencer", whose incentives are naturally very high (like her social influence), and who are typically preferred by more traditional types of advertising such as TV ads; as well as the ``ordinary-influencer''~\cite{bakshy2011}, a non-celebrity individual who is an expert in some specific topic, and thus has a relatively restricted audience, or tribe, that trust her. Second, incentives not only play their main role, i.e., encourage the seed users to endorse an advert campaign, but also, as a by-product, they incentivize users of the social media platform to become  influential in some topics by actively producing \emph{good-quality content}.  This has an obvious direct benefit for the social media platform.



\enlargethispage{\baselineskip}
\spara{Revenue maximization.} In the context of incentivized social advertising, we study the fundamental problem of revenue maximization from the host perspective: an advertiser enters into an agreement with the host to pay, following the CPE model, a fixed price $\cpe{i}$ for each engagement with ad $i$. The agreement also specifies the finite budget $B_i$ of the advertiser for the campaign for ad $i$. The host has to carefully select the seed users for the campaign: given the maximum amount $B_i$ that it can receive from the advertiser, the host must try to achieve as many engagements on the ad $i$ as possible, while spending as little as possible on the incentives for ``seed'' users. The host's task gets even more challenging by having to simultaneously accommodate multiple  campaigns by different advertisers. Moreover, for a fixed time window (e.g., 1 day, or 1 week), the host can select each user as the seed endorser for at most one ad: this constraint maintains higher credibility for the endorsements and avoids the undesirable situation where, e.g., the same sport celebrity endorses Nike and Adidas in the same time window. Therefore two ads $i$ and $j$, which are in the same topical area, naturally compete for the influential users in that area.

We show that, taking all important factors (such as topical relevance of ads, their propensity for social propagation, the topical influence of users, seed  incentives and advertiser budgets) into account, the problem of revenue maximization in incentivized social advertising corresponds to the problem of \emph{monotone submodular function maximization subject to a partition matroid constraint on the ads-to-seeds allocation, and submodular knapsack constraints on the advertisers' budgets}. This problem is NP-hard and furthermore is far more challenging than the classical influence maximization problem (IM) \cite{kempe03}  and its variants.
For this problem, we develop two natural greedy algorithms, for which we provide formal approximation guarantees. The two algorithms differ in their sensitivity to cost-effectiveness in the seed user selection:
\squishlist
  \item \emph{Cost-Agnostic Greedy Algorithm} (\CARM), which greedily chooses the seed users based on the marginal gain in the revenue, without using any information about the users' incentive costs;
  \item \emph{Cost-Sensitive Greedy Algorithm} (\CSRM), which greedily chooses the seed users based on the \emph{rate} of marginal gain in revenue per marginal gain in the advertiser's payment for each advertiser.
\squishend



Our results generalize the results of Iyer \emph{et al.}~\cite{iyer2013submodular, iyer2015submodularthesis} on submodular function maximization by $(i)$ generalizing from
a single submodular knapsack constraint to multiple submodular knapsack constraints, and $(ii)$ by handling an additional partition matroid constraint.
Our theoretical analysis leverages the notion of curvature of submodular functions.
%


Our approximation algorithms require repeatedly estimating the expected marginal gain in revenue as well in advertiser payment.
We leverage recent advances in scalable estimation of expected influence spread and devise scalable algorithms for revenue maximization in our model.

\enlargethispage{\baselineskip}
\spara{Contributions and roadmap.}
\squishlist
\item We propose \emph{incentivized} social advertising, and formulate a fundamental problem of revenue maximization from the host perspective, when the incentives paid to the seed users are determined by their demonstrated past influence in the topic of the specific ad (Section~\ref{sec:problem}).

\eat{
\note[Cigdem]{do we need to say above that the incentives are proportional to past influence? Our theory and algorithms are generic, they do not use any information regarding being "proportional" to past influence, so they are applicable to any kind of cost or payment function (as long as payment function is submodular.)}
}

\item We prove the hardness of our problem and we devise two greedy algorithms with approximation guarantees. The first (\CARM) is agnostic to users' incentives during the seed selection while the other (\CSRM) is not (Section~\ref{sec:theory}). 

\item We devise scalable versions of our approximation algorithms (Section~\ref{sec:algorithms}). Our comprehensive experimentation on real-world datasets (Section~\ref{sec:experiments}) confirms the scalability of our methods and shows that the scalable version of \CSRM consistently outperforms that of \CARM, and is far superior to natural baselines, thanks to a mindful allocation of budget on incentives. 

\squishend
Related work is discussed in Section~\ref{sec:related} while Section~\ref{sec:conclusions} concludes the paper discussing future work.

\vspace{-4mm}

\LL{\section{Problem statement}
\label{sec:problem}
}

\spara{Business model: the advertiser.} An \emph{advertiser}\footnote{\scriptsize We assume each advertiser has one ad to promote per time window, and use $i$ to refer to the $i$-th advertiser and its ad interchangeably.} $i$ enters into an agreement with the \emph{host}, the owner of the social networking platform, for an incentivized social advertising campaign on its ad. The advertiser agrees to pay the host:
\squishdesc
\item[1.] an incentive $c_i(u)$ for each seed user $u$ chosen to endorse ad $i$; we let $S_i$ denote the set of users selected to endorse ad $i$;
\item[2.] a cost-per-engagement amount $\cpe{i}$ for each user that engages with (e.g., clicks) its ad $i$.
\squishend
An advertiser $i$ has a finite budget $B_i$ that limits the amount it can spend on the campaign for its ad.
\eat{
Let $\sigma_i(S_i)$ denote the \emph{number of clicks} received by the ad $i$, when using $S_i$ as the seed set of incentivized users.
The total payment advertiser $i$ needs to make for its campaign, denoted  $\rho_i(S_i)$, is the sum of its total costs for the ad-engagements (e.g., clicks), and for incentivizing its seed users: i.e., $\rho_i(S_i) = \pi_i(S_i) + c_i(S_i)$ where $\pi_i(S_i) = \cpe{i} \cdot \sigma_i(S_i)$ and $c_i(S_i) := \sum_{u \in S_i} c_i(u)$. }

\spara{Business model: the host.} 
The host receives from advertiser $i$:

\squishdesc
\item[1.] a description of the ad $i$ (e.g., a set of keywords) which allows the host to map the ad to a distribution $\vec{\gamma_i}$ over a latent topic space (described in more detail later);
\item[2.] a commercial agreement that specifies the cost-per-engagement amount $\cpe{i}$ and the campaign budget $B_i$.
\squishend

The host is in charge of running the campaign, by selecting which users and how many to allocate as a seed set $S_i$ for each ad $i$, and by determining their incentives.
Given that these decisions must be taken \emph{before} the campaign is started, the host has to reason in terms of \emph{expectations} based on past performance.
Let $\sigma_i(S_i)$ denote the \emph{expected number of clicks} ad $i$ receives when using $S_i$ as the seed set of incentivized users. The host models
the total payment that advertiser $i$ needs to make for its campaign, denoted  $\rho_i(S_i)$, as the sum of its total costs for the \emph{expected}  ad-engagements (e.g., clicks), and for incentivizing its seed users: i.e., $\rho_i(S_i) = \pi_i(S_i) + c_i(S_i)$ where $\pi_i(S_i) = \cpe{i} \cdot \sigma_i(S_i)$ and $c_i(S_i) := \sum_{u \in S_i} c_i(u)$, where $c_i(u)$ denotes the incentive paid to a candidate seed user $u$ for ad $i$.
%
We assume $c_i(u)$ is a monotone function $f$ of the influence potential of $u$, capturing the intuition that seeds with higher expected spread cost more: i.e., $c_i(u) := f(\sigma_i(\{u\}))$.

Notice that the expected revenue of the host from the engagements to ad $i$ is just  $\pi_i(S_i)$, as the cost $c_i(S_i)$ paid by the advertiser to the host for the incentivizing influential users, is in turn paid by the host to the seeds.
In this setting, the host faces the following trade-off in trying to maximize its revenue. Intuitively, targeting influential seeds would increase the expected number of clicks, which in turn could yield a higher revenue. However, influential seeds cost more to incentivize. Since the advertiser has a fixed overall budget for its campaign, the higher seeding cost may come at the expense of reduced revenue for the host.
Finally, an added challenge is that the host has to serve many advertisers at the same time, with potentially competitive ads, i.e., ads which are very close in the topic space.

\spara{Data model, topic model, and propagation model.} The host, owns:
a \emph{directed graph} $G=(V,E)$ representing the social network, where an arc $(u,v)$ means that user $v$ follows user $u$, and thus $v$ can see $u$'s posts and may be influenced by $u$. The host also owns a \emph{topic model} for ads and users' interests, defined  by a hidden variable $Z$ that can range over $L$ latent topics. A topic distribution thus abstracts the interest pattern of a user and the relevance of an ad to those interests. More precisely, the topic model  maps each ad $i$ to a distribution $\vec{\gamma_i}$ over the latent topic space:
$$\gamma_i^z = \Pr(Z=z|i), \mbox{ with } \sum_{z = 1}^L\gamma_i^z = 1.$$

Finally, the host uses a topic-aware influence propagation model defined on the social graph $G$ and the topic model.
The propagation model governs the way in which ad impressions propagate in the social network, driven by topic-specific influence.
In this work, we adopt the \emph{Topic-aware Independent Cascade} model\footnote{\scriptsize
Note that the use of the topic-based model is orthogonal to the technical development and contributions of our work. Specifically, if we assume that the topic distributions of all ads and users are identical, the TIC model reduces to the standard IC model. The techniques and results in the paper remain intact.
} (TIC) proposed by Barbieri et al.~\cite{BarbieriBM12} which extends the standard \emph{Independent Cascade} (IC) model~\cite{kempe03}: In TIC, an ad is represented by a topic distribution, and the influence strength from user $u$ to $v$ is also topic-dependent, i.e., there is a probability $p_{u,v}^z$ for each topic $z$.
In this model, when a node $u$ clicks an ad $i$, it gets one chance of influencing each of its out-neighbors $v$ that has not clicked $i$. This event succeeds with a probability equal to the weighted average of the arc probabilities w.r.t.\ the topic distribution of ad $i$:
\begin{equation}\label{eq:tic}
	p^i_{u,v} = \sum\nolimits_{z = 1}^L\gamma_i^z \cdot p_{u,v}^z.
\end{equation}
Using this stochastic propagation model the host can determine the \emph{expected spread} $\sigma_i(S_i)$ of a given campaign for ad $i$ when using $S_i$ as seed set.
For instance, the influence value of a user $u$ for ad $i$ is defined as the expected spread of the singleton seed $\{u\}$ for the given the description for ad $i$,  under the TIC model, i.e., $\sigma_i(\{u\})$: this is the quantity that is used to determine the incentive for a candidate seed user $u$ to endorse the ad $i$.

\spara{The revenue maximization problem.}
Hereafter we assume a fixed time window (say a 24-hour period) in which the revenue maximization problem is defined.
Within this time window we have $h$ advertisers with ad description  $\vec{\gamma_i}$, cost-per-engagement $\cpe{i}$, and budget $B_i$, $i \in [h]$. We define an \emph{allocation} $\vecall$ as a vector of $h$ \emph{pairwise disjoint} sets $(S_1, \cdots, S_h) \in 2^V \times \cdots \times 2^V$, where $S_i$ is the seed set assigned to advertiser $i$ to start the ad-engagement propagation process. Within the time window, each user in the platform can be selected to be seed for at most one ad, that is, $S_i \cap S_j = \emptyset$, $i,j \in [h]$. We denote the total revenue of the host from advertisers as the sum of the  ad-specific revenues:
$$
 \pi(\vecall) =  \sum_{i \in [h]} \pi_i(S_i).
$$

\noindent Next, we formally define the revenue maximization problem for incentivized social advertising from the host perspective. Note that given an instance of the TIC model on a social graph $G$, for each ad $i$, the ad-specific influence probabilities are determined by Eq.~(\ref{eq:tic}).

\begin{problem}[\RMLong(\RM)]\label{pr:revMax}
Given a social graph $G=(V,E)$, $h$ advertisers, cost-per-engagement $\cpe{i}$ and budget $B_i$, $i \in [h]$, ad-specific influence probabilities $p^i_{u,v}$  and seed user incentive costs $c_i(u)$, $u, v\in V$, $i \in [h]$, find a feasible allocation $\vecall$ that maximizes the host's revenue:
\begin{equation*}
\begin{aligned}
& \underset{\vecall}{\text{\em maximize}}
& & \pi(\vecall) \\
& \text{\em subject to}
& & \rho_i(S_i) \le B_i, \forall i \in [h],  \\
&&& S_i \cap S_j = \emptyset , i \neq j, \forall i,j \in [h].
\end{aligned}
\end{equation*}
\end{problem}

In order to avoid degenerate problem instances, we assume that no single user incentive 
exceeds any advertiser's budget.
This ensures that every advertiser can afford at least one seed node.

\section{Hardness and Approximation}
\label{sec:theory}
{\bf Hardness}. We first show that Problem \ref{pr:revMax} (\RM) is NP-hard. We recall that a set function $f:2^U\ra \posreals$ is monotone if for $S\subset T\subseteq U$, $f(S) \le f(T)$. We define the marginal gain of an element $x$ w.r.t. $S\subset U$ as $f(x|S) := f(S\cup\{x\}) - f(S)$. A set function $f$ is submodular if for $S\subset T\subset U$ and $x\in U\setminus T$, $f(x|T) \le f(x|S)$, i.e., the marginal gains diminish with larger sets.

It is well known that the influence spread function $\sigma_i(\cdot)$ is monotone and submodular \cite{kempe03}, from which it follows that the  ad-specific revenue function $\pi_i(\cdot)$ is monotone and submodular. Finally, since the total revenue function, $\pi(\vecall) = \sum_{i \in [h]} \pi_i(S_i)$, is a non-negative linear combination of monotone and submodular functions, these properties carry over to $\pi(\vecall)$. Likewise, for each ad $i$, the payment function $\rho_i(\cdot)$ is a non-negative linear combination of two monotone and submodular functions, $\pi_i(\cdot)$ and $c_i(\cdot)$, and so is also monotone and submodular. Thus, the constraints $\rho_i(S_i) \le B_i$, $i\in[h]$, in Problem~\ref{pr:revMax} are submodular knapsack constraints. We start with our hardness result.

\begin{theorem}\label{claim:npHard}
Problem \ref{pr:revMax} (\RM) is \NPhard.
\end{theorem}

\begin{proof}
Consider the special case with one advertiser, i.e., $h = 1$. Then we have one submodular knapsack constraint and no partition matroid constraint. This corresponds to maximizing a submodular function subject to a submodular knapsack constraint, the so-called Submodular Cost Submodular Knapsack (SCSK) problem, which is known to be NP-hard~\cite{iyer2013submodular}. Since this is a special case of Problem \ref{pr:revMax}, the claim follows.
\end{proof}

Next, we characterize the constraint that the allocation $\vecall = (S_1, \cdots, S_h)$ should be composed of pairwise disjoint sets, i.e., $S_i \cap S_j = \emptyset , i \neq j, \forall i,j \in [h]$. We will make use of the following notions on matroids.
\enlargethispage{2\baselineskip}

\begin{definition}[Independence System]\label{def:IS}
A set system $(\mathcal{E},\mathcal{I})$ defined with a finite ground set $\mathcal{E}$ of elements, and a family $\mathcal{I}$ of subsets of $\mathcal{E}$ is an independence system if $\calI$ is non-empty and if it satisfies downward closure axiom, i.e.,
$X \in \mathcal{I} \wedge Y \subseteq X \rightarrow Y \in \mathcal{I}$.

\end{definition}

\begin{definition}[Matroid]\label{def:Matroid}
An independence system $(\mathcal{E}, \mathcal{I})$ is a matroid $\mathfrak{M} = (\mathcal{E}, \mathcal{I})$ if it also satisfies the augmentation axiom: i.e.,
$X \in \mathcal{I} \wedge Y \in \mathcal{I} \wedge |Y| > |X| \rightarrow \exists e \in Y \setminus X : X \cup \{e\} \in \mathcal{I}$.
\end{definition}

\begin{definition}[Partition Matroid]\label{def:PartitionMatroid}
Let $\mathcal{E}_1, \cdots, \mathcal{E}_l$ be a partition of the ground set $\mathcal{E}$ into $l$ non-empty disjoint subsets. Let $d_i$ be an integer, $0 \le d_i \le |\mathcal{E}_i|$. In a partition matroid $\mathfrak{M} = (\mathcal{E}, \mathcal{I})$, a set $X$ is defined to be independent iff, for every $i$, $1\le i\le l$, $|X \cap \mathcal{E}_i| \le d_i$. That is, $\mathcal{I} = \{X \subseteq \mathcal{E} : |X \cap \mathcal{E}_i| \le d_i, \forall i = 1,\cdots,l \}$.
\end{definition}

\begin{lemma}\label{lem:matroid}
The constraint that in an allocation $\vecall = (S_1, \cdots, S_h)$, the seed sets $S_i$ are pairwise disjoint is a partition matroid constraint over the ground set $\mathcal{E}$ of all (node, advertiser) pairs.
\end{lemma}

\begin{proof}
Given $G=(V,E)$, $\Card{V} = n$, and a set $A = \{i : i \in [h] \}$ of advertisers, let $\mathcal{E} = V \times A$ denote the ground set of all $(node, advertiser)$ pairs. Define $\mathcal{E}_u = \{ (u,i) : i \in A \}$, $u \in V$. Then the set $\{\mathcal{E}_u : \forall u \in  V \}$ forms a partition of $\mathcal{E}$ into $n$ disjoint sets, i.e., $\mathcal{E}_u \cap \mathcal{E}_v = \emptyset$, $u \ne v$, and $\bigcup_{u \in V} \mathcal{E}_u = \mathcal{E}$. Given a subset $\mathcal{X} \subseteq \mathcal{E}$, define
\begin{align*}
S_i = \{u : (u,i) \in \mathcal{X}\}.
\end{align*}
Then it is easy to see that the sets $S_i, i\in [h]$ are pairwise disjoint iff the set $\mathcal{X}$ satisfies the constraint
\begin{align*}
\mathcal{X} \cap \mathcal{E}_u \le 1, \forall u\in V.
\end{align*}

The lemma follows on noting that the set system $\mathfrak{M} = (\mathcal{E}, \mathcal{I})$, where $\mathcal{I} = \{\mathcal{X} \subseteq \mathcal{E} : |\mathcal{X}\cap\mathcal{E}_u| \le 1, \forall u \in V\}$ is actually a partition matroid.
\end{proof}

 Therefore, the \RM problem corresponds to the problem of submodular function maximization subject to a partition matroid constraint $\mathfrak{M} = (\mathcal{E}, \mathcal{I})$, \emph{and} $h$ submodular knapsack constraints.

\spara{Approximation analysis.}
Next lemma states that the constraints of the \RM problem together form an independence system defined on the ground set $\mathcal{E}$. This property will be leveraged later in developing approximation algorithms.
Given the partition matroid constraint $\mathfrak{M} = (\mathcal{E}, \mathcal{I})$, and $h$ submodular knapsack constraints, let $\mathcal{C}$ denote the family of subsets, defined on $\mathcal{E}$, that are feasible solutions to the \RM problem.

\begin{lemma}\label{lem:IS}
The system $(\calE, \calC)$ is an independence system.
\end{lemma}

\begin{proof}
 For each knapsack constraint $\rho_i(\cdot) \le B_i$, let $\mathcal{F}_i \subseteq 2^V$ denote the collection of feasible subsets of $V$, i.e.,
\begin{align*}
\mathcal{F}_i = \{ S_i \subseteq V : \rho_i(S_i) \le B_i \}.
\end{align*}
The set system $(V,\mathcal{F}_i)$ defined by the set of feasible solutions to any knapsack constraint is downward-closed, hence is an independence system. Given $\mathcal{F}_i$, $\forall i \in [h]$ and the partition matroid constraint $\mathfrak{M} = (\mathcal{E}, \mathcal{I})$, we can define the family of subsets of $\mathcal{E}$ that are feasible solutions to the \RM problem as follows:
\begin{align*}
\mathcal{C} = \left\{ \mathcal{X} : \mathcal{X} \in \mathcal{I} \text{ and } S_i \in \mathcal{F}_i, \forall i \in [h]  \right\}
\end{align*}
where $S_i = \{ u : (u,i) \in \mathcal{X} \}$. Let $\mathcal{X} \in \mathcal{C}$ and $\mathcal{X}' \subseteq \mathcal{X}$. In order to show that $\mathcal{C}$ is an independence system, it suffices to show that $\calX'\in \calC$.

Let $S_i' = \{ u : (u, i) \in \mathcal{X}' \} $, $i  \in [h]$.
Clearly, $S_i' \subseteq S_i$. As each single knapsack constraint $\rho_i(\cdot) \le B_i$ is associated with the independence system $(V, \mathcal{F}_i)$, we have $S_i' \in \mathcal{F}_i$ for any $S_i' \subseteq S_i$, $i \in [h]$.

Next, as $\mathcal{X} \in \mathcal{I}$, we have $S_i \cap S_j = \emptyset$.
Since $\mathfrak{M} = (\calE, \calI)$ is a partition matroid, by downward closure, $\calX' \in \calI$, and hence $S_i'\cap S_j' = \emptyset$, $i \ne j$.
\eat{
Given the one-to-one correspondence between $S_i'$ and $\mathcal{X}'$, and $\mathcal{X}' \in \mathcal{I}$ due to the downward closure property of the partition matroid $\mathfrak{M}$.
}
We just proved $\calX' \in \calC$, verifying that $\calC$ is an independence system.

\end{proof}
Our theoretical guarantees for our approximation algorithms to the \RM problem depend on the notion of \emph{curvature} of submodular functions. Recall that $f(j|S)$, $j\not\in S$, denotes the marginal gain $f(S\cup\{j\}) - f(S)$.
\begin{definition}[Curvature]\cite{conforti1984submodular}
Given a submodular function $f$, the total curvature $\kappa_f$ of $f$ is defined as
\begin{align*}
\kappa_f = 1 - \underset{j \in V} {\min} \dfrac{f(j|V\setminus\{j\})}{f(\{j\})},
\end{align*}
and the curvature $\kappa_f(S)$ of $f$ wrt a set $S$ is defined as
\begin{align*}
\kappa_f(S) = 1 - \underset{j \in S} {\min} \dfrac{f(j|S\setminus\{j\})}{f(\{j\})}.
\end{align*}
\end{definition}
It is easy to see that $0 \le \kappa_f = \kappa_f(V) \le 1$. Intuitively, the curvature of a function measures the deviation of $f$ from \emph{modularity}: modular functions have a curvature of $0$, and the further away $f$ is from modularity, the larger $\kappa_f$ is. Similarly, the curvature $\kappa_f(S)$ of $f$ \emph{wrt} a set $S$ reflects how much the marginal gains $f(j \mid S)$ can decrease as a function of $S$, measuring the deviation from modularity, given the context $S$.
Iyer et al.~\cite{iyer2015submodularthesis} introduced the notion of \emph{average} curvature $\hat{\kappa}_f(S)$ of $f$ wrt a set $S$ as
\begin{align*}
\hat{\kappa}_f(S) = 1 - \dfrac{\sum_{j \in S} f(j|S \setminus\{j\} )}{\sum_{j \in S} f(\{j\})},
\end{align*}
and showed the following relation between these several forms of curvature:
\begin{align*}
 0 \le \hat{\kappa}_f(S) \le \kappa_f(S)  \le \kappa_f(V) = \kappa_f \le 1.
\end{align*}
In the next subsections, we propose two greedy approximation algorithms for the \RM problem. The first of these, Cost-Agnostic Greedy Algorithm (\CARM), greedily chooses the seed users solely based on the marginal gain in the revenue, without considering seed user incentive costs. The second, Cost-Sensitive Greedy Algorithm (\CSRM), greedily chooses the seed users based on the \emph{rate} of marginal gain in revenue per marginal gain in the advertiser's payment for each advertiser.

We note that Iyer et al. \cite{iyer2013submodular, iyer2015submodularthesis} study a restricted special case of the \RM problem, referred as Submodular-Cost Submodular-Knapsack (SCSK), and propose similar cost-agnostic and cost-sensitive algorithms. Our results extend theirs in two major ways. First, we extend from a single advertiser to multiple advertisers (i.e., from
a single submodular knapsack constraint to multiple submodular knapsack constraints). Second, unlike SCSK, our \RM problem is  subject to an additional partition matroid constraint on the ads-to-seeds allocation, which naturally arises when multiple advertisers are present.

\enlargethispage{2\baselineskip}
\subsection{Cost-Agnostic Greedy Algorithm}
The Cost-Agnostic Greedy Algorithm (\CARM) for the \RM problem, whose pseudocode is provided in Algorithm~\ref{alg:CA-EARM}, chooses at each iteration a (node, advertiser) pair that provides the maximum increase in the revenue of the host.
Let $\mathcal{X}_g \subseteq \mathcal{E}$ denote the greedy solution set of (node,advertiser) pairs, returned by \CARM,  having one-to-one correspondence with the greedy allocation $\vecalloc$, i.e., $S_i = \{u : (u,i) \in \mathcal{X}_g \}$, $\forall S_i \in \vecalloc$. Let $\mathcal{X}_g^t$ denote the greedy solution after $t$ iterations of \CARM. At each iteration $t$, \CARM first finds the (node,advertiser) pair $(u^*,i^*)$ that maximizes  $\pi_i(u \mid S^{t-1}_i)$, and tests whether adding this pair to the current greedy solution $\mathcal{X}_g^{t-1}$ would violate any constraint: if $\mathcal{X}_g^{t-1} \cup \{(u^*,i^*)\}$ is feasible, the pair $(u^*,i^*)$ is added to the greedy solution as the $t$-th (node,advertiser) pair. Otherwise, $(u^*,i^*)$ is removed from the current ground set of (node,advertiser) pairs $\mathcal{E}^{t-1}$. \CARM terminates when there is no feasible (node,advertiser) pair left in the current ground set $\mathcal{E}^{t-1}$.

\begin{observation}\label{obs:curvatureTotalPi}
Being monotone and submodular, the total revenue function $\pi(\vecalloc)$ has a total curvature $\kappa_{\pi}$, given by: 
\begin{align*}
\kappa_{\pi} = 1 - \underset{(u,i) \in \mathcal{E}} {\min} \dfrac{\pi_i(u \mid V \setminus\{u\})}{\pi_i(\{u\})}.
\end{align*}
\end{observation}

\begin{proof}
Let $g : 2^{\mathcal{E}} \mapsto \mathbb{R}_{\ge 0}$ be monotone and submodular. 
Then, the total curvature $\kappa_g$ of $g$ is defined as follows:
\begin{align*}
\kappa_g = 1 - \underset{x \in \mathcal{E}} {\min} \dfrac{g(x \mid \mathcal{E} \setminus\{x\} )}{g(\{x\})},
\end{align*}
where $x = (u,i) \in \mathcal{E}$. Using the one-to-one correspondence between $\mathcal{X}_g$ and $\vecalloc$, we can alternatively formulate the \RM problem as follows:

\begin{equation*}
\begin{aligned}
& \underset{\mathcal{X} \subseteq \mathcal{E}}{\text{maximize}}
& & g(\mathcal{X}) \\
& \text{subject to}
& & \mathcal{X} \in \mathcal{C}.  \\
\end{aligned}
\end{equation*}
where $g(\mathcal{X}) = \sum_{i \in [h]} \pi_i(S_i)$ with $S_i = \{ u : (u,i) \in \mathcal{X} \}$. \\

Using this correspondence, we can rewrite $\kappa_g$ as $\kappa_{\pi}$ as follows:
\begin{align*}
\kappa_g = \kappa_{\pi} = 1 - \underset{(u,i) \in \mathcal{E}} {\min} \dfrac{\pi_i(\{u\} \mid V \setminus\{u\})}{\pi_i(\{u\})}.
\end{align*}
\end{proof}
We will make use of the following notions in our results on approximation guarantees.
\begin{definition}[Upper and lower rank]
Let $({\cal E}, {\cal C})$ be an independence system. Its \emph{upper rank} $R$ and \emph{lower rank} $r$ are defined as the cardinalities of the smallest and largest maximal independent sets:
$$r = \min \{|X| : X \in \mathcal{C} \text{ and } X \cup \{(u,i)\} \not\in \mathcal{C}, ~\forall (u,i) \not\in X\}, $$
$$R = \max \{|X| : X \in \mathcal{C} \text{ and } X \cup \{(u,i)\} \not\in \mathcal{C}, ~\forall (u,i) \not\in X\}.$$
\end{definition}

When the independence system is a matroid, $r=R$, as all maximal independent sets have the same cardinality.
\begin{theorem}\label{theo:CARM}
\CARM achieves an approximation guarantee of $\dfrac{1}{\kappa_{\pi}}\left[ 1- \left(\dfrac{R - \kappa_{\pi}}{R}\right)^{r} \right]$ to the optimum, where $\kappa_{\pi}$ is the total curvature of the total revenue function $\pi(\cdot)$, $r$ and $R$ are respectively the lower and upper rank of $(\mathcal{E}, \mathcal{C})$. \LL{This bound is tight}.
\eat{
, \emph{i.e.}, the cardinalities of the smallest and largest maximal independent sets in $\mathcal{C}$:
$$r = \min \{|X| : X \in \mathcal{C} \text{ and } X \cup \{(u,i)\} \not\in \mathcal{C}, ~\forall (u,i) \not\in X\}, $$
$$R = \max \{|X| : X \in \mathcal{C} \text{ and } X \cup \{(u,i)\} \not\in \mathcal{C}, ~\forall (u,i) \not\in X\}.$$
}
\end{theorem}

\begin{proof}
We note that the family $\mathcal{C}$ of subsets that constitute \emph{feasible} solutions to the \RM problem form an independence system defined on $\mathcal{E}$ (Lemma~\ref{lem:IS}). Given this, the approximation guarantee of \CARM directly follows from the result of Conforti et al.~\cite[Theorem 5.4]{conforti1984submodular} for submodular function maximization subject to an independence system constraint. \LL{However, the tightness does \emph{not} directly follow from the tightness result in \cite{conforti1984submodular}, which we address next.}

\LL{We now exhibit an instance to show that the bound is tight. Consider one advertiser, i.e., $h = 1$. The network is shown in Figure~\ref{fig:tightness}, where all influence probabilities are $1$. The incentive costs for nodes are as shown in the figure, while $cpe(.) = 1$. The budget is $B=7$.  It is easy to see that the lower rank is $r = 1$, corresponding to the maximal feasible seed set $S = \{b\}$, while the upper rank is $R = 2$, e.g., corresponding to  maximal feasible seed sets such as $T = \{a, c\}$. Furthermore, the total curvature is $\kappa_{\pi} = 1$. On this instance, the optimal solution is $T$ which achieves a revenue of $6$. In its first iteration, \CARM could choose $b$ as a seed. Once it does, it is forced to the solution $S = \{b\}$ as no more seeds can be added to $S$. The revenue of \CARM is $3 = \dfrac{1}{\kappa_{\pi}}\left[ 1- \left(\dfrac{R - \kappa_{\pi}}{R}\right)^{r} \right]OPT = \frac{1}{2} \cdot 6$. }
 \end{proof}

\begin{figure}[t!]
\vspace{-4mm}
\centering{\includegraphics[width=.39\textwidth, height=.12\textheight]{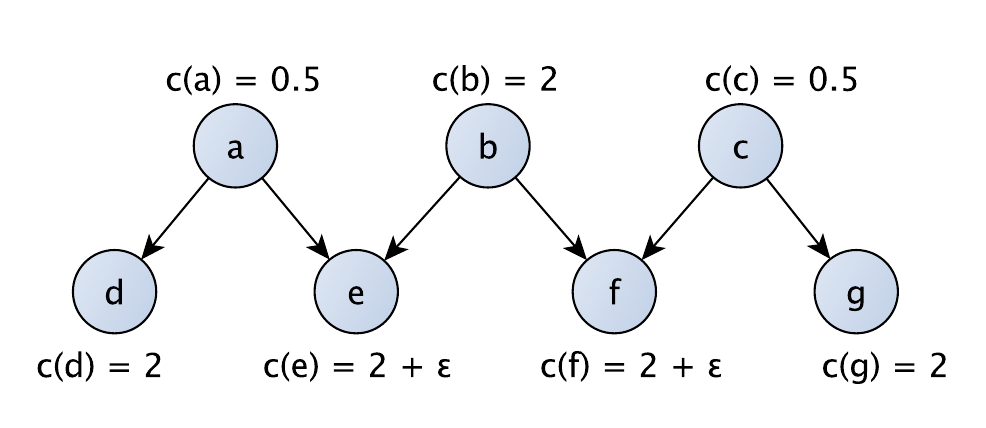}}
\vspace{-4mm}
\caption{Instance illustrating tightness of bound in Theorem~\ref{theo:CARM}.}
\label{fig:tightness}
\vspace{2mm}
\end{figure}

\noindent \LL{\textbf{Discussion.} We next discuss the significance and the meaning of the bound in Theorem~\ref{theo:CARM}. 
Notice that when there is just one advertiser, TIC reduces to IC. Even for this simple setting, the bound on \CARM is tight.
By a simple rearrangement of the terms, we have:
$$\dfrac{1}{\kappa_{\pi}}\left[ 1- \left(\dfrac{R - \kappa_{\pi}}{R}\right)^{r} \right] \ge \dfrac{1}{\kappa_{\pi}} \left( 1 - e^{ - \kappa_{\pi} \cdot \dfrac{r}{R}} \right).
$$
Clearly, the cost-agnostic approximation bound improves as $\dfrac{r}{R}$ approaches $1$, achieving the best possible value when $r = R$. As a special case, the cost-agnostic approximation further improves when the independence system $(\mathcal{E}, \mathcal{C})$ is a matroid since for a matroid $r = R$ always holds: e.g., consider the standard IM problem~\cite{kempe03} which corresponds to submodular function maximization subject to a uniform matroid. Here, $\pi(\cdot) = \sigma(\cdot)$.
Then the 
approximation guarantee becomes $\dfrac{1}{\kappa_{\pi}} \left( 1 - e^{ - \kappa_{\pi}} \right)$, providing a slight improvement over the usual $(1 - 1/e)$-approximation, thanks to the curvature term $\kappa_{\pi}$.\footnote{\small \LL{Note that $\kappa_{\pi} \le 1$ always. Hence, the extent of improvement increases as the total curvature $\kappa_{\pi}$ decreases.}} This remark is also valid for budgeted influence maximization~\cite{LeskovecKDD07} with uniform seed costs. }
%
\LL{For more general instances of the problem, the guarantee depends on the characteristics of the instance, specifically, the lower and upper ranks and the curvature. This kind of instance dependent bound is characteristic of submodular function maximization over an independence system~\cite{korte78,conforti1984submodular}. Specifically for the RM problem, given its constraints, the values of $r$ and $R$ are dictated by the values of $h$ payment functions over all feasible allocations. For instance, given our assumption that every advertiser can afford at least one seed, we always have $r \ge h$. The worst-case value $r = h$ corresponds to the case in which each advertiser $i$ is allocated a single seed node $u_i$ whose payment $\rho_i(u_i)$ exhausts its budget $B_i$. Similarly for $R$, without using any particular assumption on $B_i$, $\forall i \in [h]$, we always have $R \le \text{min}(n, \sum_{i \in [h]} \left\lfloor B_i / cpe(i) \right\rfloor)$.}
\CA{Notice also that:
\begin{align}
\dfrac{1}{\kappa_{\pi}}\left[ 1- \left(\dfrac{R - \kappa_{\pi}}{R}\right)^{r} \right] &= \dfrac{1}{\kappa_{\pi}}\left[ 1- \left(1 - \dfrac{ \kappa_{\pi}}{R}\right)^{r} \right] \\
\ge  \dfrac{1}{\kappa_{\pi}} \left[ 1- \left(1 - \dfrac{ \kappa_{\pi}}{R}\right) \right] &= \dfrac{1}{\kappa_{\pi}} \dfrac{ \kappa_{\pi}}{R} = \dfrac{1}{R}
\end{align}
Hence, the worst-case approximation is always bounded by $1/R$.}


\begin{algorithm}[!h!]
 \caption{\CARM}
\label{alg:CA-EARM}
{\small
\SetKwInOut{Input}{Input}
\SetKwInOut{Output}{Output}
\SetKwComment{tcp}{//}{}
\Input{$G=(V,E)$, $B_i$, $\cpe{i}$, $\vec{\gamma}_i, \forall i \in [h]$, $c_i(u), \forall i \in [h], \forall u \in V$}
 \Output{$\vecalloc = (S_1, \cdots, S_h)$}
} 
{\small
$t \leftarrow 1$, $\mathcal{E}^0 \leftarrow \mathcal{E}$, $\mathcal{X}_g^0 \leftarrow \emptyset$ \\
$S^0_i \leftarrow \emptyset$, $\forall i \in [h]$ \\
}
\While{$\mathcal{E}^{t-1} \neq \emptyset$} {
$(u^*,i^*) \leftarrow \argmax_{(u,i) \in \mathcal{E}^{t-1}}  \pi_i(u \mid S^{t-1}_i)$ \\ \label{line:greedyCriteria}
\uIf{$(\mathcal{X}_g^{t-1} \cup \{(u^*,i^*)\}) \in \mathcal{C}$} {
			$S_{i^*}^{t} \leftarrow S_{i^*}^{t-1} \cup \{ u^*\} $ \\
			$S_j^t \leftarrow S_j^{t-1}$, $\forall j \neq i^*$ \\
			$\mathcal{X}_g^{t} \leftarrow \mathcal{X}_g^{t-1} \cup \{(u^*, i^*)\}$ \\
			$\mathcal{E}^t \leftarrow \mathcal{E}^{t-1} \setminus \{(u^*, i^*)\}$ \\
			$t \leftarrow t + 1$ \\
}
\uElse{
$\mathcal{E}^{t-1} \leftarrow \mathcal{E}^{t-1} \setminus \{(u^*, i^*)\} $ \\
}
}
$S_i \leftarrow S^{t-1}_i$, $\forall i \in [h]$ \\
{\bf return} $\vecalloc = (S_1, \cdots, S_h)$
\end{algorithm}

\vspace{2mm} \subsection{Cost-Sensitive Greedy Algorithm}
The Cost-sensitive greedy algorithm (\CSRM) for the \RM problem is similar to \CARM. The main difference is that at each iteration $t$, \CSRM first finds the (node,advertiser) pair $(u^*,i^*)$ that maximizes $\dfrac{\pi_i(u \mid S^{t-1}_i)}{\rho_i(u \mid S^{t-1}_i)} $, and tests whether the addition of this pair to the current greedy solution set $\mathcal{X}_g^{t-1}$ would violate any matroid or knapsack independence constraint: if the addition is feasible, 
the pair $(u^*,i^*)$ is added to the greedy solution as the $t$-th (node,advertiser) pair. Otherwise, $(u^*,i^*)$ is removed from the current ground set 
$\mathcal{E}^{t-1}$.
\CSRM terminates when there is no (node,advertiser) pair left in the current ground set $\mathcal{E}^{t-1}$. \CSRM can be obtained by simply replacing Line~\ref{line:greedyCriteria} of Algorithm~\ref{alg:CA-EARM} with
$$ (u^*,i^*) \leftarrow \underset{(u,i) \in \mathcal{E}^{t-1}} \argmax \dfrac{\pi_i(u \mid S^{t-1}_i)}{\rho_i(u \mid S^{t-1}_i)}. $$


\begin{theorem}\label{theo:cs-earm1}
\CSRM achieves an approximation guarantee of
$$1 - \dfrac{R \cdot \rho_{max} }{R \cdot \rho_{max} + (1 - \underset{i \in [h]} {\text{max }} {\kappa_{\rho_{i}}}) \cdot \rho_{min}} $$
 to the optimum where $R$ is the upper rank of $(\mathcal{E}, \mathcal{C})$, $\kappa_{\rho_{i}}$ is the total curvature of $\rho_i(\cdot)$, $\forall i \in [h]$,  $\rho_{max} := \underset{(u,i) \in \mathcal{E}} {\text{max }} \rho_{i}(u) $ and $\rho_{min} := \underset{(u,i) \in \mathcal{E}} {\text{min }} \rho_{i}(u) $ are respectively the maximum and minimum singleton payments over all (node, advertiser) pairs.
\end{theorem}
\begin{proof}
We use $\vec{S^*} = (S_1^*, ..., S_h^*)$ and $\vecalloc = (S_1, ..., S_h)$ to denote the optimal and greedy allocations respectively, and $\calX^*$ and $\calX_g$ to denote the corresponding solution sets. Specifically, $S_i^* = \{u : (u,i) \in \mathcal{X}^*\}$, and $S_i = \{u : (u,i) \in \mathcal{X}_g \}$. We denote by $\calX_g^t$ the result of the greedy solution after $t$ iterations. Let $K = |\mathcal{X}_g|$ denote the size of the greedy solution. Thus, $\calX_g = \calX_g^K$. By submodularity and monotonicity:
\begin{small}
$$\pi(\vec{S^*}) \le \pi(\vecalloc)  + \sum_{(u,i) \in \mathcal{X}^* \setminus \mathcal{X}_g} \pi_i(u \mid S_i) \le \pi(\vecalloc)  + \sum_{(u,i) \in \mathcal{X}^*} \pi_i(u \mid S_i).$$
\end{small}
At each iteration $t$, the greedy algorithm first finds the (node, advertiser) pair $(u^*,i^*) \leftarrow \underset{(u,i) \in \mathcal{E}^{t-1}} \argmax \dfrac{\pi_i(u \mid S^{t-1}_i)}{\rho_i(u \mid S^{t-1}_i)}$, and tests whether the addition of this pair to the current greedy solution set $\mathcal{X}_g^{t-1}$ would violate any independence constraint. If $(u^*,i^*)$ is feasible, i.e., if $\mathcal{X}_g^{t-1} \cup \{(u^*,i^*)\} \in \mathcal{C}$, then the pair $(u^*,i^*)$ is added to the greedy solution as the $t$-th (node, advertiser) pair; otherwise, $(u^*,i^*)$ is removed from the current ground set $\mathcal{E}^{t-1}$. In what follows, for clarity, we use the notation $(u_t, i_t)$ to denote the (node, advertiser) pair that is \emph{successfully} added by the greedy algorithm to $\mathcal{X}_g^{t-1}$ in iteration $t$.

Let $U^{t}$ denote the set of (node, advertiser) pairs that the greedy algorithm \emph{tested} for possible addition to the greedy solution in the first $(t+1)$ iterations \emph{before} the addition of the $(t+1)$-st pair $(u_{t+1}, i_{t+1})$ into $\mathcal{X}_g^{t}$. Thus, $U^t \setminus U^{t-1}$ includes the $t$-th pair $(u_t, i_t)$ that was successfully added to $\mathcal{X}_g^{t-1}$, as well as all the pairs that were tested for addition into $\mathcal{X}_g^{t}$ but failed the independence test. Thus, $\forall (u,i) \in U^t \setminus U^{t-1}$, we have
$\dfrac{\pi_{i}(u \mid S_{i}^{t})}{\rho_{i}(u \mid S_{i}^{t})} \ge \dfrac{\pi_{i_{t+1}}(u_{t+1} \mid S_{i_{t+1}}^{t})}{\rho_{i_{t+1}}(u_{t+1} \mid S_{i_{t+1}}^{t})}$, since they were tested for addition to $\mathcal{X}^t_g$ before $(u_{t+1}, i_{t+1})$, but failed the independence test. For all $(u,i) \in U^t \setminus U^{t-1}$, we have $\dfrac{\pi_{i}(u \mid S_{i}^{t-1})}{\rho_{i}(u \mid S_{i}^{t-1})} \le \dfrac{\pi_{i_t}(u_t \mid S_{i_t}^{t-1})}{\rho_{i_t}(u_t \mid S_{i_t}^{t-1})}$.
since they were not good enough to be added to $\mathcal{X}_g^{t-1}$ as the $t$-th pair. Note that, the greedy algorithm terminates when there is no feasible pair left in the ground set. Hence after $K$ iterations, $\mathcal{E}^{K}$ contains only the \emph{infeasible} pairs that violate some matroid or knapsack constraint. Thus, we have $\mathcal{X}^* = \bigcup_{t=1}^{K} [\mathcal{X}^* \cap (U^t \setminus U^{t-1}) ]$. Let $ \mathcal{U}^*_t := \mathcal{X}^* \cap (U^t \setminus U^{t-1})$. Notice that $\calX^* = \bigcup_{t=1}^K \calU_t^*$. Then, we have:
\begin{align*}
\pi(\vec{S^*}) &\le \pi(\vecalloc)  + \sum_{(u,i) \in \mathcal{X}^*} \pi_i(u \mid S_i) \\
 &= \pi(\vecalloc)  + \sum_{t=1}^{K} \sum_{(u,i) \in \mathcal{U}^*_t} \pi_i(u \mid S_i) \\
&\le \pi(\vecalloc)  + \sum_{t=1}^{K} \sum_{(u,i) \in \mathcal{U}^*_t} \dfrac{\pi_{i_t}(u_t \mid S_{i_t}^{t-1})}{\rho_{i_t}(u_t \mid S_{i_t}^{t-1})} \cdot \rho_i(u \mid S_i^{t-1}).
\end{align*}
The last inequality is due to the fact that
$\forall (u,i) \in \mathcal{U}^*_t$:
\begin{align*}
\pi_i(u \mid S_i) \le \pi_i(u \mid S_{i}^{t-1}) \le \dfrac{\pi_{i_t}(u_t \mid S_{i_t}^{t-1})}{\rho_{i_t}(u_t \mid S_{i_t}^{t-1})} \cdot \rho_{i}(u \mid S_{i}^{t-1}),
\end{align*}
where the first inequality follows from submodularity and the second follows from the greedy choice of (node, advertiser) pairs.
Continuing, we have:
 \begin{align}
 \label{eq:toBeContd}
\pi(\vec{S^*}) &\le \pi(\vecalloc)  + \sum_{t=1}^{K} \sum_{(u,i) \in \mathcal{U}^*_t} \dfrac{\pi_{i_t}(u_t \mid S_{i_t}^{t-1})}{\rho_{i_t}(u_t \mid S_{i_t}^{t-1})} \cdot \rho_i(u \mid S_i^{t-1}) \nonumber \\
&= \pi(\vecalloc)  + \sum_{t=1}^{K} \dfrac{\pi_{i_t}(u_t \mid S_{i_t}^{t-1})}{\rho_{i_t}(u_t \mid S_{i_t}^{t-1})} \sum_{(u,i) \in \mathcal{U}^*_t} \rho_i(u \mid S_i^{t-1})  \nonumber  \\
&\le \pi(\vecalloc)  + \sum_{t=1}^{K} \dfrac{\pi_{i_t}(u_t \mid S_{i_t}^{t-1})}{\rho_{i_t}(u_t \mid S_{i_t}^{t-1})} \cdot \sum_{t=1}^{K} \sum_{(u,i) \in \mathcal{U}^*_t} \rho_i(u) \nonumber \\
&= \pi(\vecalloc)  + \sum_{t=1}^{K} \dfrac{\pi_{i_t}(u_t \mid S_{i_t}^{t-1})}{\rho_{i_t}(u_t \mid S_{i_t}^{t-1})} \cdot  \sum_{(u,i) \in \mathcal{X}^*} \rho_i(u) \nonumber \\
&\le \pi(\vecalloc)  + \pi(\vecalloc) \cdot \dfrac{R \cdot \underset{(u,i) \in \mathcal{X}^*} {\max} \rho_i(u)  }{\underset{t \in [1,K]} {\text{min }} \rho_{i_t}(u_t \mid S_{i_t}^{t-1})}
\end{align}
where the last inequality follows from the fact that $\pi(\vecalloc) = \sum_{t=1}^{K} \pi_{i_t}(u_t \mid S_{i_t}^{t-1})$ and $\Card{\mathcal{X}^*} \le R$ since $\mathcal{X}^* \in \mathcal{C}$. Let $(u_{t_m}, i_{t_m}) := \underset{t \in [1,K]} {\argmin} \rho_{i_t}(u_t \mid S_{i_t}^{t-1})$ and let $(u_{min}, i_{min}) := \underset{(u,i) \in \mathcal{E}} {\argmin} \rho_i(u \mid V \setminus \{u\}) $. Being monotone and submodular, each $\rho_i(\cdot)$ has the total curvature $\kappa_{\rho_i} = 1 - \underset{u \in V} {\min} \dfrac{\rho_i(u \mid V\setminus \{u\})}{\rho_i(u)}$.
Hence, for $\rho_{i_{min}}(\cdot)$, we have:
\begin{align} \label{eq:hedem2}
1 -  \kappa_{\rho_{i_{min}}} &= \underset{u \in V} {\min} \dfrac{\rho_{i_{min}}(u \mid V\setminus \{u\})}{\rho_{i_{min}}(u)} \le \dfrac{\rho_{i_{min}}(u_{min} \mid V \setminus \{u_{min}\})}{\rho_{i_{min}}(u_{min})},
\end{align}
where the inequality above follows from the definition of total curvature. Then, using submodularity and Eq.\ref{eq:hedem2}, we obtain:
\begin{align}
\label{eq:anotherEquation}
\underset{t \in [1,K]} {\text{min }} \rho_{i_t}(u_t \mid S_{i_t}^{t-1}) &= \rho_{i_{t_m}}(u_{t_m} \mid S_{i_{t_m}}^{t_m - 1}) \nonumber \\
&\ge \rho_{i_{t_m}}(u_{t_m} \mid V \setminus \{u_{t_m}\}) \nonumber \\
&\ge \underset{(u,i) \in \mathcal{E}} {\min} \rho_i(u \mid V \setminus \{u\}) \nonumber \\
&= \rho_{i_{min}}(u_{min} \mid V \setminus \{u_{min}\}) \nonumber \\
&\ge (1 -  \kappa_{\rho_{i_{min}}}) \cdot {\rho_{i_{min}}(u_{min})} \nonumber \\
&\ge (1 - \underset{i \in [h]} {\text{max }} {\kappa_{\rho_{i}}}) \cdot  \underset{(u,i) \in \mathcal{E}} {\text{min }} {\rho_{i}(u)}.
\end{align}
Continuing from where we left in Eq.\ref{eq:toBeContd} and using Eq.\ref{eq:anotherEquation}, we have:
\begin{align}
\label{eq:garantiFinal}
\pi(\vec{S^*}) &\le \pi(\vecalloc) + \pi(\vecalloc) \cdot \dfrac{R \cdot \underset{(u,i) \in \mathcal{X}^*} {\max} \rho_i(u)  }{\underset{t \in [1,K]} {\text{min }} \rho_{i_t}(u_t \mid S_{i_t}^{t-1})} \nonumber \\
&\le \pi(\vecalloc) \cdot \left( 1+  \dfrac{R \cdot \underset{(u,i) \in \mathcal{E}} {\max} \rho_i(u) }{ (1 - \underset{i \in [h]} {\text{max }} {\kappa_{\rho_{i}}}) \cdot  \underset{(u,i) \in \mathcal{E}} {\text{min }} {\rho_{i}(u)} }  \right) \nonumber \\
&= \pi(\vecalloc) \cdot \left( 1+  \dfrac{R \cdot \rho_{max} }{ (1 - \underset{i \in [h]} {\text{max }} {\kappa_{\rho_{i}}}) \cdot \rho_{min} }  \right)
\end{align}
Rearranging the terms we obtain:
\begin{align*}
\pi(\vecalloc) &\ge \pi(\vec{S^*}) \cdot \dfrac{  (1 - \underset{i \in [h]} {\text{max }} {\kappa_{\rho_{i}}}) \cdot  \rho_{min}  } {  (1 - \underset{i \in [h]} {\text{max }} {\kappa_{\rho_{i}}}) \cdot \rho_{min} +  R \cdot \rho_{max}} \\
&= \pi(\vec{S^*}) \cdot \left(1 - \dfrac{R \cdot \rho_{max} }{R \cdot \rho_{max} + (1 - \underset{i \in [h]} {\text{max }} {\kappa_{\rho_{i}}}) \cdot \rho_{min}} \right).
\end{align*}
\end{proof}

\noindent \LL{\textbf{Discussion.} We next discuss the significance and the meaning of the bounds. Notice that the value of the cost-sensitive approximation bound improves as the ratio $\dfrac{\rho_{max}}{\rho_{min}}$ decreases, as Eq.~\ref{eq:garantiFinal} shows. Since $\rho_{max} \le \underset{i \in [h]} {\text{min }} B_i$, we can see that as the value of $\rho_{max}$ decreases, intuitively $r$ would increase, for the corresponding maximal independent set of minimum size could pack more seeds under the knapsack constraints. Similarly, if the value of $\rho_{min}$ increases, $R$ would decrease since the corresponding maximal independent set of maximum size could pack fewer seeds under the knapsack constraints. Thus, intuitively as $\dfrac{\rho_{max}}{\rho_{min}}$ decreases, $\dfrac{r}{R}$ would increase. When this happens, both cost-agnostic and cost-sensitive approximations improve.}

\LL{At one extreme, when ${\kappa_{\rho_{i}}} = 0, \forall i \in [h]$, i.e., when $\rho_i(\cdot)$ is modular $\forall i \in [h]$, we have linear knapsack constraints. Thus, Theorem~\ref{theo:CARM} and Theorem~\ref{theo:cs-earm1} respectively provide cost-agnostic and cost-sensitive approximation guarantees for the \emph{Budgeted Influence Maximization} problem~\cite{LeskovecKDD07, nguyen2013budgeted} for the case of multiple advertisers, with an additional matroid constraint.}
%
%
\LL{At the other extreme, when $\underset{i \in [h]} {\text{max }} {\kappa_{\rho_{i}}} = 1$, which is the case for totally normalized and saturated functions (\emph{e.g.}, matroid rank functions), the approximation guarantee of \CSRM is unbounded, i.e., it becomes degenerate. This is similar to the result of \cite{iyer2015submodularthesis} for the SCSK problem whose cost-sensitive approximation guarantee becomes unbounded.
\eat{
due to the average curvature term $\hat{\kappa}_{\rho}(S^*)$ that it employs where $S^*$ is the optimal solution of the SCSK problem.
}
Nevertheless, combining the results of the cost-agnostic and cost-sensitive cases, we can obtain a bounded approximation. }

\LL{On the other hand, while \CARM always has a bounded worst-case guarantee, our experiments show that \CSRM empirically obtains higher revenue\footnote{It remains open whether the approximation bound for \CSRM is tight. Interestingly, on the instance (Fig.~\ref{fig:tightness}) used in the proof of Theorem\ref{theo:CARM}, \CSRM obtains the optimal solution $T = \{a, c\}$.}.
}


%

\section{Scalable Algorithms}
\label{sec:algorithms}
While Algorithms \CARM and \CSRM provide approximation guarantees, their efficient implementation is a challenge, as both of them require a large number of influence spread computations: in each iteration $t$, for each advertiser $i$ and each node $u \in V \setminus S^{t-1}_i$, the algorithms need to compute $\pi_i(u \mid S^{t-1}_i)$ and $\pi_i(u \mid S^{t-1}_i) / \rho_i(u \mid S^{t-1}_i)$, respectively.

Computing the exact influence spread $\sigma(S)$ of a given seed set $S$ under the IC model is \SPhard~\cite{ChenWW10}, and this hardness carries over to the TIC model. In recent years, significant advances have been made in efficiently estimating $\sigma(S)$. A natural question is whether they can be adapted to our setting, an issue we address next.

\eat{A common practice is to use Monte Carlo (MC) simulations~\cite{kempe03}.
However, accurate estimation requires a large number of MC simulations, which is prohibitively expensive and will not scale.
Thus, to make \CARM and \CSRM scalable, we need an alternative approach.}

\subsection{Scalable Influence Spread Estimation}
\eat{
In the influence maximization literature, considerable effort has been devoted to developing scalable approximation algorithms.
Recently, Borgs et al.~\cite{borgs14} 
introduced the idea of sampling \emph{``reverse-reachable''} (RR) sets in the graph for the efficient estimation of influence spread, and proposed a quasi-linear time randomized algorithm.}

Tang et al.~\cite{tang14} 
proposed a near-linear time randomized algorithm for influence maximization, called {\em Two-phase Influence Maximization (TIM)}, building on the notion of \emph{``reverse-reachable''} (RR) sets proposed by Borgs et al.~\cite{borgs14}. Random RR-sets are critical in the efficient estimation of influence spread.
Tang et al.~\cite{tang2015influence} subsequently proposed an algorithm called IMM that improves upon TIM by tightening the lower bound on the number of random RR-sets required to estimate influence with high probability. The difference between TIM and IMM is that the lower bound used by TIM ensures that the number of random RR-sets it uses is sufficient to estimate the spread of \emph{any} seed set of a given size $s$. By contrast, IMM uses a lower bound that is tailored for the seed that is greedily selected by the algorithm. Nguyen et al.~\cite{NguyenTD16}, adapting ideas from TIM~\cite{tang14}, and the sequential sampling design proposed by Dagum \emph{et al.}~\cite{dagum2000optimal}, proposed an algorithm called SSA that provides significant run-time improvement over TIM and IMM.

These algorithms are designed for the basic influence maximization problem and hence require knowing the number of seeds as input. In our problem, the number of seeds is not fixed, but is dynamic and depends on the budget and partition matroid constraints. Thus a direct application of these algorithms is not possible.

Aslay et al.~\cite{AslayLB0L15} recently proposed a technique for efficient seed selection for IM when the number of seeds required is not predetermined but can change dynamically. \CA{However, their technique cannot handle the presence of seed user incentives which, in our setting, directly affects the number of seeds required to solve the RM problem}. In this section, we derive inspiration from their technique. First, though we note that for \CARM, in each iteration, for each advertiser, we need to find a feasible node that yields the maximum marginal gain in revenue, and hence the maximum marginal spread. By contrast, in \CSRM, we need to find the node that yields the maximum \emph{rate} of marginal revenue per marginal gain in payment, i.e., $\pi_i(u \mid S_i^{t-1}) / \rho_i(u \mid S_i^{t-1})$.

To find such node $u_i^{t}$
we must compute $\sigma_i(v | S^{t-1}_i)$, $\forall v : (v,i) \in \mathcal{E}^{t-1}$: notice that node $u_i^{t}$ might even correspond to the node that has the \emph{minimum} marginal gain in influence spread for iteration $t$.
{\sl Thus, any scalable realization of \CSRM should be capable of working as an influence spread oracle that can efficiently compute} $\pi_i(u \mid S_i^{t-1}) / \rho_i(u \mid S_i^{t-1})$ {\sl for all} $u \in \{v : (v,i) \in \mathcal{E}^{t-1}\}$.

Among the state-of-the-art IM algorithms~\cite{tang14, tang2015influence, NguyenTD16}, only TIM~\cite{tang14} can be adapted to serve as an influence oracle. For a given set size $s$, the derivation of the number of random RR-sets that TIM uses is done such that the influence spread of {\sl any set of at most $s$ nodes can be accurately estimated}.
On the other hand, even though IMM~\cite{tang2015influence} and SSA~\cite{NguyenTD16} provide significant run-time improvements over TIM, they inherently cannot perform this estimation task accurately: the sizes of the random RR-sets sample that these algorithms use are tuned just for accurately estimating the influence spread of \emph{only} the approximate greedy solutions; the sample sizes used are inadequate for estimating the spread of arbitrary seed sets of a given size. Thus, we choose to extend TIM to devise scalable realizations of \CARM and \CSRM, namely, \fastca and \fastcs. Next, we describe how to extend the ideas of RR-sets sampling and TIM's sample size determination technique to obtain scalable approximation algorithms for the RM problem: \fastca  and \fastcs.

\enlargethispage{\baselineskip}
\CA{\subsection{Scalable Revenue Maximization}}
For the scalable estimation of influence spread, in this section we devise \fastca and \fastcs, scalable realizations of \CARM and \CSRM, based on the notion of Reverse-Reachable sets~\cite{borgs14} and adapt the sample size determination procedure employed by TIM~\cite{tang14} to achieve a certain estimation accuracy with high confidence.

\spara{Reverse-Reachable (RR) sets~\cite{borgs14}.} Under the IC model, a random RR-set $R$ from $G$ is generated as follows. First, for every edge $(u,v) \in E$, remove it from $G$ w.p.\ $1-p_{u,v}$: this generates a possible world (deterministic graph) $X$. Second, pick a \emph{target} node $w$ uniformly at random from $V$. Then, $R$ consists of the nodes that can reach $w$ in $X$. For a sufficient sample $\RR$ of random RR-sets, the fraction $F_{\RR}(S)$ of $\RR$ covered by $S$ is an unbiased estimator of $\sigma(S)$, i.e., $\sigma(S) = \E[ n \cdot F_\RR(S)]$.

\spara{Sample Size Determination of \CC{TIM}~\cite{tang14}.} \CC{Let $\CN{\RR_i}$ be a collection of $\theta_i$ random RR-sets. Given any seed set size $s_i$ and $\varepsilon > 0$, define $\CC{L_{i}(s_i,\varepsilon)}$ to be:}
\begin{align}\label{eqn:timLB}
\CC{L_{i}(s_i,\varepsilon)} = (8 + 2 \varepsilon) n \cdot \dfrac{\ell \log n + \log \binom{n}{\CN{s_i}} + \log 2}{\CC{OPT_{i,s_i}} \cdot \varepsilon^{2}},
\end{align}
where $\ell > 0, \varepsilon > 0$ and \CC{$OPT_{i,s_i} = \underset{S \subseteq V, |S| \le s_i} {\max~} \sigma_i(S)$}. Let $\theta_i$ be a number no less than $\CC{L_{i}(s_i,\varepsilon)}$. Then, for any seed set $S$ with $\CN{|S| \leq s_i}$, the following inequality holds w.p.\ at least $1 - n^{- \ell} / \binom{n}{\CN{s_i}}$:
\CN{
\begin{align}
\label{eq:Lemma3}
\left| n \cdot F_{\RR_i}(S_i) - \sigma_i(S_i) \right| < \dfrac{\varepsilon}{2} \cdot OPT_{\CC{i,s_i}}.
\end{align}
}


\CN{
\spara{Estimated Payments and Budget Feasibility.\footnote{\footnotesize{We would like to thank to Kai Han and Jing Tang for bringing the budget feasibility issue into our attention, which we address in this section.}}} Let $\vecstilde = (\stilde_1, \cdots, \stilde_h)$ denote the approximately greedy solution that \fastca (resp. \fastcs) returns. Since the algorithm operates on the estimation of influence spread, the revenue and payment computed for each advertiser $i$ will also be estimations of the actual revenue and payment for seed set $\stilde_i$. Let $\spi_i(\stilde_i) =  \cpe{i} \cdot n \cdot \isfrac{\stilde_i}$ and $\srho_i(\stilde_i) = c_i(\stilde_i) + \spi_i(\stilde_i)$ denote the estimated revenue and estimated payment for advertiser $i$, respectively. As \fastca (resp. \fastcs) performs budget feasibility check on the estimated payments, it is possible to encounter scenarios in which $\srho_i(\stilde_i) \le B_i$ while $\rho_i(\stilde_i) > B_i$. Thus, to ensure that the approximate greedy allocation results in actual payments that do not violate any budget constraints with high probability, one could consider to use a refined budget $\rbudget{i} < B_i$, for each advertiser $i$, by taking into account the error introduced by spread estimation. Next, we provide details on how to set $\rbudget{i}$ so that $\stilde_i$ is budget feasible with high probability.}

\CN{First, notice that, following Eq.\ref{eq:Lemma3}, we have $\sigma_i(\stilde_i)  \le n \cdot \isfrac{\stilde_i} + \frac{\varepsilon}{2} \cdot OPT_{\CC{i,s_i}}$. Thus, to ensure that $c_i(\stilde_i) + \cpe{i} \cdot \sigma_i(\stilde_i) \le B_i$, w.h.p., we need to have:  
\begin{align*}
c_i(\stilde_i) + \cpe{i} \cdot \left(n \cdot \isfrac{\stilde_i} + \frac{\varepsilon}{2} \cdot OPT_{\CC{i,s_i}}\right) \le B_i
\end{align*}
which implies that the budget constraint on the estimated payment $\srho_i(\stilde_i)$ should be refined as:
\begin{align}
\srho_i(\stilde_i) \le B_i - \cpe{i} \cdot \frac{\varepsilon}{2} \cdot OPT_{\CC{i,s_i}}.
\end{align}
While using a refined budget of  $B_i - \cpe{i} \cdot \frac{\varepsilon}{2} \cdot OPT_{\CC{i,s_i}}$ would ensure w.h.p. that $\rho_i(\stilde_i) \le  B_i$, such refinement requires to compute $OPT_{\CC{i,s_i}}$ which is unknown and \NPhard to compute. To circumvent this difficulty, one could consider an upper bound $\CC{\etb{i}{s_i}}$ on $OPT_{\CC{i,s_i}}$ so that 
\begin{align*}
\rbudget{i} &= B_i -  \cpe{i} \cdot \frac{\varepsilon}{2} \cdot \CC{\etb{i}{s_i}}\\
&\le B_i -  \cpe{i} \cdot \frac{\varepsilon}{2} \cdot \CC{OPT_{i,s_i}}. 
\end{align*}}

\CC{Following \cite{tang18}, an upper bound $\CC{\etb{i}{s_i}}$ on $\CC{OPT_{i,s_i}}$ can be obtained as follows.}

\begin{lemma}[Restated from Lemma $4.3$~\cite{tang18}]
\label{lemma:ubEta}
Let $\RR_i$ be a sample of $\theta_i$ RR-sets, such that, $\theta_i \ge \CC{L_{i}(s_i,\varepsilon)}$, and let $\CC{\tilde{A}_i} \subseteq V$, $|\CC{\tilde{A}_i}| = s_i$ denote the greedy solution to maximum coverage problem on the sample $\RR_i$. Define $\CC{\etb{i}{s_i}}$ to be:
\begin{align}\label{eq:imUB}
\CC{\etb{i}{s_i}} := \left(\sqrt{\frac{\theta_i \cdot \isfrac{\CC{\tilde{A}_i}}}{1 - 1/e} +\frac{\ln{ n^{\ell}}}{2}} + \sqrt{\frac{\ln{n^{\ell}}}{2}} \right)^2 \cdot \frac{n}{\theta_i}
\end{align}
Then, we have: 
\CC{\begin{align*}
\text{Pr}\left[OPT_{i,s_i}  \le \etb{i}{s_i} \right] \ge 1 - n^{\ell}.
\end{align*}}
\end{lemma}

\CC{Following Lemma~\ref{lemma:ubEta}, for a given seed set size $s_i$, we can define $\rbudget{i}$ for $i$ as: 
\begin{align}
\label{eq:refBudget}
\rbudget{i} = B_i - \cpe{i} \cdot \frac{\varepsilon}{2} \cdot \etb{i}{s_i}. 
\end{align}}

\spara{Latent Seed Set Size Estimation.} \CC{The derivation of the sufficient sample size, depicted in Eq.~\ref{eqn:timLB}, requires the number of seeds as input for each $i$, which is not available for \RM problem. Let $s^*_i = |S^*_i|$ denote the true number of seeds that the optimal allocation would assign to $i$. From the advertisers' budgets, there is no obvious way to determine $s^*_i$ for each $i$. This poses a challenge as the required number of RR-sets ($\theta_i$) for advertiser $i$ depends on $s^*_i$.}

\CC{To circumvent this difficulty, one can use a safe upper bound $\sib{i} = \left\lceil \frac{B_i}{\rho^i_{min}} \right\rceil$ on $s^*_i$, where $\rho^i_{min}$ is the minimum singleton payment for $i$ so that, by using a sample of at least $L_{i}(\sib{i},\varepsilon)$ RR-sets, we can quantify how the approximation guarantee of  \fastca (resp, \fastcs) deteriorate from the guarantee of \CARM (resp., \CSRM) as a function of the estimation accuracy that the sample size ensures for all seed sets of size at most $\sib{i}$ (Eq.\ref{eq:Lemma3}). However, when $\rho_{min}$ is very small w.r.t. $B_i$, a direct application of TIM's sample size derivation technique for $\sib{i}$ seeds could result in a large estimation error $\dfrac{\varepsilon}{2} \cdot OPT_{i,{\sib{i}}}$, due to $\sib{i}$ being a very loose upper bound on $s^*_i$. \CN{Such large estimation error could translate to working with a refined budget $\rbudget{i}$ that is very small w.r.t. $B_i$, resulting in greatly under-utilizing the budget for the sake of budget feasibility.} Now, we explain how to derive a sample size that can estimate the spread of any seed set of size at most $\sib{i}$ while using a more stringent estimation error $\dfrac{\varepsilon}{2} \cdot OPT_{i,\tilde{s}_i}$ with $\tilde{s}_i < \sib{i}$, where $\tilde{s}_i$ is the latent seed set size estimation obtained during the execution of \fastca (resp., \fastcs) as we will explain next.} 


\CC{\begin{lemma}
\label{lemma:newSS}
Let $\CN{\RR_i}$ be a collection of $\theta_i$ random RR-sets. Given $\sib{i}$, $\tilde{s}_i$, and $\varepsilon > 0$, define $\CC{L_{i}(\sib{i}, \tilde{s}_i, \varepsilon)}$ to be:
\begin{align}\label{eqn:timL2}
L_{i}(\sib{i}, \tilde{s}_i, \varepsilon) = (8 \lambda  + 2 \varepsilon) n \cdot \dfrac{\ell \log n + \log \binom{n}{\CN{\sib{i}}} + \log 2}{\CC{OPT_{i,\tilde{s}_i}} \cdot \varepsilon^{2}},
\end{align}
where $\ell > 0, \varepsilon > 0$, \CC{$OPT_{i,s} = \underset{S \subseteq V, |S| \le s} {\max~} \sigma_i(S)$}, for any integer $s$, and $\lambda =\frac{OPT_{i, \sib{i}}}{OPT_{i,\tilde{s}_i}}$. Let $\theta_i$ be a number no less than $L_i(\sib{i}, \tilde{s}_i, \varepsilon) $. Then, for any seed set $S$ with $|S| \leq \sib{i}$, the following inequality holds w.p.\ at least $1 - n^{- \ell} / \binom{n}{\sib{i}}$:
\begin{align}
\label{eq:newSS}
\left| n \cdot F_{\RR_i}(S) - \sigma_i(S) \right| < \dfrac{\varepsilon}{2} \cdot OPT_{i,\tilde{s}_i}.
\end{align}
\end{lemma}}

\CC{\begin{proof}
Let $S$ be any seed set of size at most $\sib{i}$ and let $\tau_i$ denote the probability that $S$ overlaps with a random RR set, i.e., 
$$\tau_i = \E[F_{\RR_i}(S)] = \frac{\sigma_i(S)}{n}.$$
Then, we have:
\begin{align}
\label{eq:newS2}
& \prob{\left| n \cdot F_{\RR_i}(S) - \sigma_i(S) \right| < \dfrac{\varepsilon}{2} \cdot OPT_{i,\tilde{s}_i}} \nonumber \\
&= \prob{\left| \theta_i \cdot F_{\RR_i}(S) - \tau_i \theta_i \right| < \dfrac{\varepsilon \theta_i}{2n} \cdot OPT_{i,\tilde{s}_i}} \nonumber \\
&= \prob{\left| \theta_i \cdot F_{\RR_i}(S) - \tau_i \theta_i \right| < \dfrac{\varepsilon \cdot OPT_{i,\tilde{s}_i}}{2n \tau_i} \cdot \tau_i \theta_i}.
\end{align}
Letting $\delta = \dfrac{\varepsilon \cdot OPT_{i,\tilde{s}_i}}{2n \tau_i}$, by Chernoff bounds, we have:
\begin{align*}
&\text{r.h.s. of Eq.\ref{eq:newS2} } < 2\exp{\left(- \dfrac{\delta^2}{2+\delta} \cdot \tau_i \theta_i  \right)} \\
&= 2\exp{\left(- \dfrac{\varepsilon^2 \cdot OPT^2_{i,\tilde{s}_i}}{8n^2\tau_i + 2\varepsilon n \cdot  OPT_{i,\tilde{s}_i}} \cdot \theta_i \right)} \\
&<  2\exp{\left(- \dfrac{\varepsilon^2 \cdot OPT^2_{i,\tilde{s}_i}}{8n \cdot OPT_{i, \sib{i}} + 2 \varepsilon n \cdot  OPT_{i,\tilde{s}_i}} \cdot \theta_i \right)} \\
&= 2\exp{\left(- \dfrac{\varepsilon^2 \cdot OPT_{i,\tilde{s}_i}}{8n \cdot \frac{OPT_{i, \sib{i}}}{OPT_{i,\tilde{s}_i}} + 2 \varepsilon n} \cdot \theta_i \right)} 
\end{align*}
where the last \emph{inequality} follows from the fact that $\tau_i \le OPT_{i, \sib{i}}$. Finally, we obtain the lower bound on $\theta_i$ by solving 
\begin{align*}
2\exp{\left(- \dfrac{\varepsilon^2 \cdot OPT_{i,\tilde{s}_i}}{8n \frac{OPT_{i, \sib{i}}}{OPT_{i,\tilde{s}_i}} + 2 \varepsilon n} \cdot \theta_i \right)}  \le \dfrac{n^{-\ell}}{\binom{n}{\sib{i}}}.
\end{align*}
\end{proof}}

\CC{An upper bound on the $\lambda$ term required for the sample size derivation in Eq.~\ref{eqn:timL2} can be obtained by using an upper bound on $OPT_{i, \sib{i}}$, as given by Lemma~\ref{lemma:ubEta}, and  a lower bound on $OPT_{i,\tilde{s}_i}$ by using the lower bounding technique provided in \cite{tang14} for TIM's sample size derivation (Eq.~\ref{eqn:timLB}).}

\eat{We now explain the ``latent seed set size estimation" procedure which first makes an initial guess at $s^*_i$, and then iteratively revises the estimated value, until no more seeds are needed, while concurrently selecting seeds and allocating them to advertisers. For ease of exposition, let us first consider a single advertiser $i$. We start with an initial  estimated value for $s^*_i$, denoted by $\tilde{s}_i^1$, and use it to obtain a corresponding sample size \CN{${\theta}_i^1 = L_{i}(\tilde{s}_i^1,\varepsilon)$ using Eq.~\ref{eqn:timLB}}, an upper bound $\CC{\etb{i}{\tilde{s_i}^1}}$ using Eq.~\ref{eq:imUB}, and a refined budget ${\rbudget{i}}^1$ using Eq.~\ref{eq:refBudget}.}

We now explain the ``latent seed set size estimation" procedure which first makes an initial guess at the true number of seeds required to maximize cost-agnostic (cost-sensitive) revenue and then iteratively revises the estimated value, until no more seeds are needed, while concurrently selecting seeds and allocating them to advertisers. For ease of exposition, let us first consider a single advertiser $i$. We start with an initial  estimate, denoted by $\tilde{s}_i^1$, and use it to obtain a corresponding sample size \CN{${\theta}_i^1 = L_{i}(\tilde{s}_i^1,\varepsilon)$ using Eq.~\ref{eqn:timLB}}, an upper bound $\CC{\etb{i}{\tilde{s_i}^1}}$ using Eq.~\ref{eq:imUB}, and a refined budget ${\rbudget{i}}^1$ using Eq.~\ref{eq:refBudget}. As it is \SPhard to compute $\rho^i_{min}$, we also compute in this iteration a safe upper bound $\sib{i}$ from  $$\sib{i} = \left\lceil \frac{B_i}{\srho^i_{min} + \cpe{i} \cdot \frac{\varepsilon}{2} \cdot \etb{i}{\tilde{s}_i}} \right\rceil$$
where $\srho^i_{min} = \underset{u \in V}{\min~} c_i(u) + cpe(i) \cdot n \cdot F_{\RR_i}(u)$. \CN{At iteration $t>1$, we compute the sample size from ${\theta}_i^t  = L_{i}(\sib{i}, \tilde{s}_i^1, \varepsilon)$}, and if ${\theta}_i^t > {\theta}_i^{t-1}$, we will need to sample additional $({\theta}_i^t - {\theta}_i^{t-1})$ RR-sets, and use all RR-sets sampled up to this iteration to select $(\tilde{s}_i^t - \tilde{s}_i^{t-1})$ additional seeds \CN{into the seed set $\stilde_i$ of advertiser $i$}, \CC{while revising the upper bound  $\CC{\etb{i}{\tilde{s_i}^t}}$ and the corresponding refined budget $\rbudget{i}^t$}. After adding those seeds, \CN{if the current payment estimate $\srho_i(\stilde_i)$} is still less than $\rbudget{i}^t$, more seeds can be assigned to advertiser $i$. Thus, we will need another iteration and we further revise our estimation of $s^*_i$. The new value, $\tilde{s}_i^{t+1}$, is obtained as follows:

\eat{The estimation of the latent seed set size required by \fastca and \fastcs can be obtained as follows: for
 Let $s_i$ be the true number of seeds required to maximize the \eat{ cost-agnostic (cost-sensitive)} revenue for advertiser $i$, i.e., $\CC{s_i = |S^*_i|}$. We do not know $s_i$ and we estimate it in successive iterations as $\tilde{s}_i^t$. \CN{Thus, we start with an initial  estimated value for $s_i$, denoted by $\tilde{s_i}^1$, and use it to obtain a corresponding sample size ${\theta}_i^1$ using Eq.~\ref{eqn:timLB}, an upper bound $\CC{\etb{i}{\tilde{s_i}^1}}$ using Eq.~\ref{eq:imUB}, and a refined budget ${\rbudget{i}}^1$ using Eq.~\ref{eq:refBudget}}. If \CN{at iteration t},  ${\theta}_i^t > {\theta}_i^{t-1}$, we will need to sample additional $({\theta}_i^t - {\theta}_i^{t-1})$ RR-sets, and use all RR-sets sampled up to this iteration to select $(\tilde{s}_i^t - \tilde{s}_i^{t-1})$ additional seeds \CN{into the seed set $\stilde_i$ of advertiser $i$}, \CC{while revising the upper bound  $\CC{\etb{i}{\tilde{s_i}^t}}$ and the corresponding refined budget $\rbudget{i}^t$}. After adding those seeds, \CN{if the current payment estimate $\srho_i(\stilde_i)$} is still less than $\rbudget{i}^t$, more seeds can be assigned to advertiser $i$. Thus, we will need another iteration and we further revise our estimation of $s_i$. The new value, $\tilde{s}_i^{t+1}$, is obtained as follows:}

\CN{
\begin{align}\label{eq:latentCA}
\tilde{s}_i^{t+1} \gets \tilde{s}_i^t + \left\lfloor \dfrac{\rbudget{i}^t - \srho_i(\stilde_i)}{c_i^{max} + \cpe{i} \cdot (n \cdot F_{\RR_i}^{max} + \frac{\varepsilon}{2} \cdot \CC{\etb{i}{\tilde{s}_i^t}})}  \right\rfloor
\end{align} \enlargethispage{\baselineskip}
}
where $c_i^{max} := \underset{v \in V} {\max}~ c_i(v)$ is the maximum seed user incentive cost for advertiser $i$, and $F_{\RR_i}^{max} := \underset{u \in V \setminus \stilde_i} {\max}F_{\RR_i}(u)$. This ensures we do not overestimate as future seeds have diminishing marginal gains, thanks to submodularity, and incentives bounded by $c_i^{max}$.



\IncMargin{1em}
\begin{algorithm}[t!]
\caption{\fastcs}
\label{alg:fastCS}
\Indm
{\small
\SetKwInOut{Input}{Input}
\SetKwInOut{Output}{Output}
\SetKwComment{tcp}{//}{}
\Input{$G=(V,E)$, $B_i$, $\cpe{i}$, $\vec{\gamma}_i, \forall i \in [h]$, $c_i(u), \forall i \in [h], \forall u \in V$}
 \Output{\CN{$\vecstilde = (\stilde_1, \ldots, \stilde_h)$}}
}
\Indp
{\small
\ForEach{$j = 1, 2, \ldots, h$} {
$\CN{\stilde_j} \gets \emptyset$; $Q_j \gets \emptyset$; \tcp{\small a priority queue}
\CN{$\tilde{s}_j \gets 1$}; $\theta_j \gets L_{\CC{j}}(\CN{\tilde{s}_j}, \varepsilon)$; $\RR_j \gets \mathsf{Sample}(G, \gamma_j,\theta_j)$\;
\CN{$\sib{j} \gets \left\lceil \frac{B_j}{\srho^j_{min} + \cpe{j} \cdot \frac{\varepsilon}{2} \cdot \etb{j}{\tilde{s}_j}} \right\rceil$}\;
\CN{$\rbudget{j} \gets B_j -  \cpe{i} \cdot \frac{\varepsilon}{2} \cdot \CC{\etb{j}{\tilde{s}_j}}$}\;
$\text{assigned}[u] \leftarrow \text{false}, \forall u \in V$\;
}
\BlankLine
\While{true} {
  \ForEach{$j = 1, 2, \ldots, h$} {
  	  $(v_j, cov_j(v_j)) \gets \mathsf{SelectBestCSNode}(\RR_j)$ (Alg~\ref{alg:rrBestCSNode}) \label{line:greedySelectBest}
 	  $F_{\RR_j}(v_j) \gets cov_j(v_j) / \theta_j $\;
     $\CN{\spi_j(\stilde_j \cup \{v_j\}) \gets \spi_j(\stilde_j) + \cpe{j} \cdot n \cdot F_{\RR_j}(v_j)}$\; 	
  }
  $i \gets \argmax_{j=1}^h \dfrac{\spi_j(v_j | \stilde_j)}{\srho_j(v_j | \stilde_j)}$ subject to: $\srho_j(\stilde_j \cup \{v_j\}) \le \CN{\rbudget{j}} \; \wedge \; \text{assigned}[v_j] = \text{false} $ \; \label{line:greedyCriterCS}
\If{$i \neq \mathbf{NULL}$} {
	$\stilde_i \gets \stilde_i \cup \{v_i\}$\;
	$\text{assigned}[v_i] = \text{true}$\;
         $Q_i.\mathsf{insert}(v_i, cov_i(v_i)) $\;
	$\RR_i \gets \RR_i \setminus \{R \mid v_i \in R \; \wedge \; R \in \RR_i\} $;
}
//{\tt remove RR-sets that are covered}\;
\lElse {
	{\bf return} //{\tt all advertisers exhausted;}
}
\If{$ \left\vert{\stilde_i}\right\vert = \tilde{s}_i$} {
	\CN{$\tilde{s}_i \gets \tilde{s}_i + \left\lfloor \frac{\rbudget{i}  - \srho_i(\stilde_i)}{c_i^{max} + \cpe{i} \cdot (n \cdot F_{\RR_i}^{max} + \frac{\varepsilon}{2} \cdot \etb{i}{\tilde{s}_i})} \right\rfloor$}\;
  	$\RR_i \gets \RR_i \cup \mathsf{Sample}(G, \gamma_i, \max\{0, L_{\CC{i}}(\tilde{s}_i,\varepsilon) - \theta_i \}$\; 	
  	\CC{$\theta_i \gets \max\{L_{\CC{i}}(\CN{\sib{i}}, \tilde{s}_i,\varepsilon), \theta_{i}\}$}\;
  	$\spi_i(\stilde_i) \gets$ $\mathsf{UpdateEstimates}$($\RR_i$, $\theta_i$, $\stilde_i$, $Q_i$)\;
  	\CN{$\rbudget{i} \gets B_i -  \cpe{i} \cdot \frac{\varepsilon}{2} \cdot \CC{\etb{i}{\tilde{s}_i}}$}\;
//{\tt revise estimates to reflect newly added RR-sets}\;
  	$\srho_i(\stilde_i) \gets$ $\spi_i(\stilde_i) + c_i(\stilde_i)$;
}
}
}
\end{algorithm}

\begin{algorithm}[t!]
\caption{UpdateEstimates($\RR_i$, $\theta_i$, $\stilde_i$, $Q_i$)}
\label{alg:rrUpdateEst}
\Indm
{\small
\SetKwInOut{Input}{Input}
\SetKwInOut{Output}{Output}
\SetKwComment{tcp}{//}{}
\Output{$\spi_i(\stilde_i)$}
}
\Indp
{ \small
$\spi_i(\stilde_i) \gets 0 $ \;
\For{$j = 0, \ldots, |\stilde_i| -1$} {
	$(v,cov_i(v)) \gets Q_i[j] $ \;
	$cov_i'(v) \gets \left\vert{\{R \mid v \in R, R \in \RR_i\}}\right\vert $\;
	$Q_i.\mathsf{insert}(v, cov_i(v) + cov_i'(v))$\;
	$\spi_i(\stilde_i) \gets \cpe{i} \cdot n \cdot ((cov_i(v) + cov_i'(v)) / \theta_i) $; //{\tt update coverage of existing seeds w.r.t. new RR-sets added to collection.}
}
}
\end{algorithm}

\begin{algorithm}[t!]
\caption{SelectBestCANode($\RR_j$)}
\label{alg:rrBestCANode}
\Indm
{\small
\SetKwInOut{Input}{Input}
\SetKwInOut{Output}{Output}
\SetKwComment{tcp}{//}{}
\Output{$(u, cov_j(u))$}
}
\Indp
{ \small
    $u \gets \argmax_{v\in V} |{\{R \mid v \in R \; \wedge \; R \in \RR_j\}}| $
\hspace*{12ex} subject to: $\text{assigned}[v] = \text{false}$\; \label{algo-selectCA:line1}
    $cov_j(u) \gets |{\{R \mid u \in R \; \wedge \; R \in \RR_j\}}| $;
//{\tt find best cost-agnostic seed for ad $j$ as well as its coverage.}
}
\end{algorithm}

\begin{algorithm}[t!]
\caption{SelectBestCSNode($\RR_j$)}
\label{alg:rrBestCSNode}
\Indm
{\small
\SetKwInOut{Input}{Input}
\SetKwInOut{Output}{Output}
\SetKwComment{tcp}{//}{}
\Output{$(u, cov_j(u))$}
}
\Indp
{ \small
    $u \gets \argmax_{v\in V} \dfrac{|{\{R \mid v \in R \; \wedge \; R \in \RR_j\}}|}{c_j(v)} $
\hspace*{12ex} subject to: $\text{assigned}[v] = \text{false}$\; \label{algo-selectCS:line1}
    $cov_j(u) \gets |{\{R \mid u \in R \; \wedge \; R \in \RR_j\}}| $;
//{\tt find best cost-sensitive seed for ad $j$ as well as its coverage.}
}
\end{algorithm}
\DecMargin{1em}


\CA{While the core logic of \fastcs (resp. \fastca) is still based on the greedy seed selection outlined for \CSRM (resp. \CARM),  \fastcs (resp. \fastca) uses random RR-sets samples for the scalable estimation of influence spread.} \CA{Since \fastca and \fastcs are very similar, differing only in their greedy seed selection criteria, we only provide the pseudocode of \fastcs (Algorithm~\ref{alg:fastCS}).} \CA{Algorithm \fastcs works as follows. For every advertiser $j$, we initially set the latent seed set size \CN{$\tilde{s}_j = 1$} (a conservative but safe estimate), create a sample $\RR_j$ of $\theta_j = L_{\CC{j}}(\tilde{s}_j, \varepsilon)$ RR-sets, \CN{compute the refined budget $\rbudget{j}$ for $\tilde{s}_j$}, and the safe upper bound $\sib{j}$ (lines 1 -- \CN{6}). In the main loop, we follow the greedy selection logic of \CSRM. That is, in each round, we first invoke Algorithm~\ref{alg:rrBestCSNode} to find an unassigned candidate node $v_j$ that has the largest coverage-to-cost ratio~\footnote{\scriptsize Following the definition of \CN{$\srho_j(\cdot)$} as a function of \CN{$\spi_j(\cdot)$}, the node with the largest rate of marginal gain in revenue per marginal gain in payment for a given ad $j$ corresponds to the node $u$ with the largest coverage-to-cost ratio for ad $j$.} for each advertiser $j$ whose budget is not yet exhausted. Then, we select, among these (node,advertiser) pairs, the feasible pair $(v_i, i)$ that has the largest rate of marginal gain in revenue per marginal gain in payment and add it to the solution set, and remove from $\RR_i$ the RR-sets that are covered by node $v_i$ (lines \CN{10 -- 15}). While doing so, whenever \CN{$|\stilde_i|  = \tilde{s}_i$}, we update the latent seed set size $\tilde{s}_i$ using Eq.~\ref{eq:latentCA}, \CN{hence $\rbudget{i}$},  and sample \CC{$\max\{0, L_{\CC{i}}(\CN{\sib{i}}, \tilde{s}_i,\varepsilon) - \theta_i \}$} additional RR-sets into $\RR_i$. Note that, after adding additional RR-sets, we update the influence spread estimation of current $\stilde_i$ w.r.t. the updated sample $\RR_i$ by invoking Algorithm~\ref{alg:rrUpdateEst} to ensure that future marginal gain estimations are accurate (line \CN{22}). The main loop executes until the budget of each advertiser is exhausted or no more eligible seed can be found.}


For \fastca, there are only two differences. First, line~\ref{line:greedySelectBest} of Algorithm~\ref{alg:fastCS} is replaced by $$(v_j, cov_j(v_j)) \gets \mathsf{SelectBestCANode}(\RR_j) \;\; \mbox{(Algorithm~\ref{alg:rrBestCANode})}.$$ Second, line~\ref{line:greedyCriterCS} of Algorithm~\ref{alg:fastCS} is replaced by \begin{align*}i \leftarrow \argmax_{j=1}^h \pi_j(v_j | \stilde_j) \; \mbox{ subject to: } \rho_j(\stilde_j \cup \{v_j\}) \le B_j \; \\ \wedge \; \text{assigned}[v_j] = \text{false}.\end{align*}


\enlargethispage{\baselineskip}
\spara{Deterioration of approximation guarantees.} \LL{Since \fastca and \fastcs 
use random RR-sets for the accurate estimation of $\sigma_i(\cdot), \forall i \in [h]$, their approximation guarantees slightly deteriorate from the ones of \CARM and \CSRM (see Theorems \ref{theo:CARM} and \ref{theo:cs-earm1}). Such deterioration is common to all the state-of-the-art IM algorithms~\cite{borgs14, tang14, tang2015influence, NguyenTD16} that similarly use random RR-sets for influence spread estimation. Our next result provides the deteriorated approximation guarantees for \fastca and \fastcs. }
\CA{\begin{theorem}\label{theo:deterioratedGaranti}
\CC{W.p. at least $1-n^{-\ell}$, \fastca (resp. \fastcs) returns a solution $\vecstilde = (\stilde_1, \ldots, \stilde_h)$ that satisfies}
\begin{align*}
\pi(\vec{\tilde{S}}) &\ge \pi(\vec{S^*}) \cdot \beta  -  \sum_{i \in [h]} cpe(i) \cdot  \varepsilon \cdot OPT_{\tilde{s}_i}.
\end{align*}
where $\vecsstar = (\sstar_1, \ldots, \sstar_h)$ is the optimal allocation, \eat{$\vecstilde = (\stilde_1, \ldots, \stilde_h)$ is the approximate greedy solution that \fastca (resp. \fastcs) returns,} \CC{$\tilde{s}_i$ is the final latent seed set size estimated for each $i$ upon termination of \fastca (resp. \fastcs)}, and $\beta$ is the approximation guarantee given in Theorem \ref{theo:CARM} (resp. Theorem \ref{theo:cs-earm1}).
\end{theorem}}
\begin{proof}
\CC{Let $\vecssample = (\ssample_1, \ldots, \ssample_h)$ denote the optimal solution to \RM problem on the sample \CC{with refined budget constraints}, i.e., the feasible allocation that maximizes $\sum_{i \in [h]} \spi_i(S_i)$ subject to \CC{$\srho_i(S_i) \le \rbudget{i}$, $\forall i \in [h]$}. Since $\vecstilde$ is the cost-agnostic (resp., cost-sensitive) greedy solution to \RM on the sample, we have:
\begin{align}
\label{eq:equationrm1}
\sum_i \cpe{i} \cdot n \cdot \isfrac{\stilde_i} \ge \beta \cdot \sum_i \cpe{i} \cdot n \cdot \isfrac{\ssample_i}.
\end{align} 
Given that $\vecssample$ is the optimal solution to solving \RM on the sample, we also have: 
\begin{align}
\label{eq:equationrm2}
\sum_i \cpe{i} \cdot n \cdot \isfrac{\ssample_i} \ge \sum_i \cpe{i} \cdot n \cdot \isfrac{\sstar_i}.
\end{align}}

Furthermore, it follows from Lemma~\ref{lemma:newSS} that, for any set $S$ of at most $\sib{i}$ seeds, we have $\left| n\cdot F_{\RR_i}(S) - \sigma_i(S) \right| \ge \dfrac{\varepsilon}{2} \cdot OPT_{i, \tilde{s}_i}$ w.p. at most $\frac{n^{-\ell}}{\binom{n}{\sib{i}}}$. Notice that, we also have $|\sstar_i| \le \sib{i}$ by definition. Thus, by using Eqs.\ref{eq:equationrm1} and \ref{eq:equationrm2} and a union bound over all $\binom{n}{\sib{i}}$ estimations, w.p. at least $1-n^{-\ell}$ we have: 


\CN{\begin{align*}
&\sum_i \cpe{i} \cdot \sigma_i(\stilde_i) \\
&\ge \sum_i \cpe{i} \cdot \left(n \cdot \isfrac{\stilde_i} - \frac{\varepsilon}{2} \cdot OPT_{i,\tilde{s}_i} \right) \\
&= \sum_i \cpe{i} \cdot n \cdot \isfrac{\stilde_i} - \sum_i \cpe{i} \cdot \frac{\varepsilon}{2} \cdot OPT_{i,\tilde{s}_i} \\
&\ge \beta \cdot \sum_i \cpe{i} \cdot n \cdot \isfrac{\ssample_i} - \sum_i \cpe{i} \cdot \frac{\varepsilon}{2} \cdot OPT_{i,\tilde{s}_i} \\
&\ge \beta \cdot \sum_i \cpe{i} \cdot n \cdot \isfrac{\sstar_i} - \sum_i \cpe{i} \cdot \frac{\varepsilon}{2} \cdot OPT_{i,\tilde{s}_i} \\
&\ge \beta \cdot \sum_i \cpe{i} \cdot \CC{\left(\sigma_i(\sstar_i) - \frac{\varepsilon}{2} \cdot OPT_{i,\tilde{s}_i}\right)} \\
&- \sum_i \cpe{i}\cdot \frac{\varepsilon}{2} \cdot OPT_{i,\tilde{s}_i} \\
&\ge \beta \cdot \pi(\vecsstar) - \sum_{i \in [h]} \cpe{i} \cdot  \varepsilon \cdot OPT_{i,\tilde{s}_i},
\end{align*}}
\ravi{where the last inequality follows upon noting that $\beta < 1$.}  
\end{proof}

As a corollary to Theorem~\ref{theo:deterioratedGaranti}, Lemma~\ref{lemma:ubEta} and Lemma~\ref{lemma:newSS}, the following result is immediate. 

\CC{\begin{theorem}\label{theo:finalResult}
W.p. at least $1-n^{-\ell}$, \fastca (resp. \fastcs) returns an approximate greedy solution $\vecstilde = (\stilde_1, \ldots, \stilde_h)$ that is budget feasible, i.e., $\rho_i(\stilde_i) \le B_i$, for all $i$, and achieves an approximation that satisfies 
\begin{align*}
\pi(\vec{\tilde{S}}) &\ge \pi(\vec{S^*}) \cdot \beta  -  \sum_{i \in [h]} cpe(i) \cdot  \varepsilon \cdot OPT_{\tilde{s}_i}.
\end{align*}
$\beta$ is the approximation guarantee given in Theorem \ref{theo:CARM} (resp. Theorem \ref{theo:cs-earm1}).
\end{theorem}}

\section{Experiments}
\label{sec:experiments}
We conducted extensive experiments to evaluate $(i)$\ the quality of our proposed algorithms, measured by the revenue achieved vis \`{a} vis the incentives paid to seed users,
and $(ii)$\ the efficiency and scalability of the algorithms w.r.t. advertiser  budgets, which indirectly control the number of seeds required, and w.r.t. the number of advertisers, which effectively controls the size of the graph. 
All experiments were run on a 64-bit OpenSuSE Linux server with Intel Xeon 2.90GHz CPU and 264GB memory.
As a preview, our largest configuration is \livej with 20 ads, which effectively yields a graph with $69M \times 20 \approx 1.4B$ edges; this is comparable with  \cite{tang14}, whose largest dataset has 1.5B edges.

\smallskip\noindent\textbf{Data.}
Our experiments were conducted on four real-world social networks, whose basic statistics are summarized in Table~\ref{table:dataset}.
We used \flix and \epi  for quality experiments and \dblp and \livej for scalability experiments.
\flix is from a social movie-rating website (\url{http://www.flixster.com/}),  which contains movie ratings by users along with timestamps.
\LL{We use the topic-aware influence probabilities and the item-specific topic distributions provided by Barbieri et al.~\cite{BarbieriBM12}, who learned the probabilities using MLE for the TIC model, with $L=10$ latent topics. We set the default number of advertisers $h=10$ and used five of the learned topic distributions from the provided \flix dataset, in such a way that every two ads are in pure competition, i.e., have the same topic distribution, with probability $0.91$ in one randomly selected latent topic, and $0.01$ in all others. This way, among $h=10$ ads, every two ads are in pure competition with each other while having a completely different topic distribution than the rest, representing a diverse marketplace of ads. \epi is a who-trusts-whom network taken from a consumer review website (\url{http://www.epinions.com/}). Likewise, we set $h=10$ and use the Weighted-Cascade model~\cite{kempe03}, where $p_{u,v}^i = 1/{|N^{in}(v)|}$ for all ads $i$. Notice that this corresponds to $L=1$ topic for \epi dataset, hence, all the ads are in pure competition.} 

\begin{table}[t!]
\vspace{-3mm}

\small
\centering
\caption{Statistics of network datasets.\label{table:dataset}}

\begin{tabular}{|c | c | c | c | c| }
\hline
& \flix & \epi  & \dblp & \livej  \\ \hline
\#nodes & 30K & 76K & 317K  & 4.8M \\ \hline
\#edges & 425K & 509K & 1.05M & 69M  \\ \hline
type & directed & directed & undirected & directed \\ \hline
\end{tabular}
\end{table}

\begin{table}[t!]
\small
\centering
\caption{Advertiser budgets and cost-per-engagement values.\label{table:cpe}}
	
\begin{tabular}{|c | c | c | c | c | c | c|}
		\hline
		 &  \multicolumn{3}{|c|}{ Budgets} & \multicolumn{3}{|c|} {CPEs} \\
 		 \hline
		 Dataset & mean & max  & min & mean & max & min  \\ \hline
		\flix & 10.1K  & 20K & 6K  & 1.5 & 2 & 1  \\ \hline
		\epi & 8.5K & 12K  & 6K  & 1.5 & 2 & 1 \\ \hline
	\end{tabular}
\vspace{1mm}
\end{table}

For scalability experiments, we used two large networks\footnote{\scriptsize Available at \url{http://snap.stanford.edu/}.} \dblp and \livej.
\dblp is a co-authorship graph (undirected) where nodes represent authors and there is an edge between two nodes if they have co-authored a paper indexed by DBLP. We direct all edges in both directions. \livej is an online blogging site where users can declare which other users are their friends.
In all datasets, advertiser budgets and CPEs were chosen in such a way that the total number of seeds required for all ads to meet their budgets is less  than $n$. This ensures that no ad is assigned an empty seed set.
For lack of space, instead of enumerating all CPEs and budgets, we give a statistical summary in Table~\ref{table:cpe}.
The same information for \dblp and \livej in provided later.

\spara{Seed incentive models.}
\LL{In order to understand how the algorithms perform w.r.t.\ different seed user incentive assignments, we used \AR{four} different methods that directly control the range between the minimum and maximum singleton payments:
\squishlisttight
\item Linear incentives: proportional to the ad-specific singleton influence spread of the nodes, i.e., $c_i(u) = \alpha \cdot \sigma_i(\{u\}), \forall u \in V, i \in [h]$,
\item Constant incentives: the average of the ad-specific total linear seed user incentives, i.e., $c_i(u) = \alpha \cdot \dfrac{\sum_{v \in V}  \sigma_i(\{v\})}{n}, \forall u \in V, i \in [h]$,
\item Sublinear incentives: obtained by taking the logarithm of the ad-specific singleton influence spread of the nodes, i.e., $c_i(u) = \alpha \cdot \log(\sigma_i(\{u\})), \forall u \in V, i \in [h]$,
\item \AR{Superlinear incentives: obtained by using the squared ad-specific singleton influence spread of the nodes, i.e., $c_i(u) = \alpha \cdot \left(\sigma_i(\{u\})\right)^2, \forall u \in V, i \in [h]$},
\squishend
where $\alpha > 0$ denotes a fixed amount in dollar cents set by the host, which controls how expensive the seed user incentives are.
}

\begin{figure}[t!]
\vspace{-4mm}
\begin{tabular}{ccc}
\vspace{-4mm}\includegraphics[width=.24\textwidth]{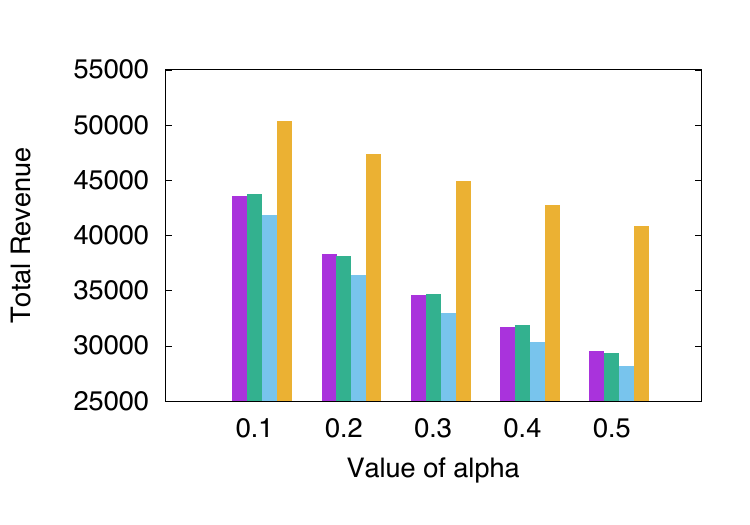}&
	\hspace{-2mm}\includegraphics[width=.24\textwidth]{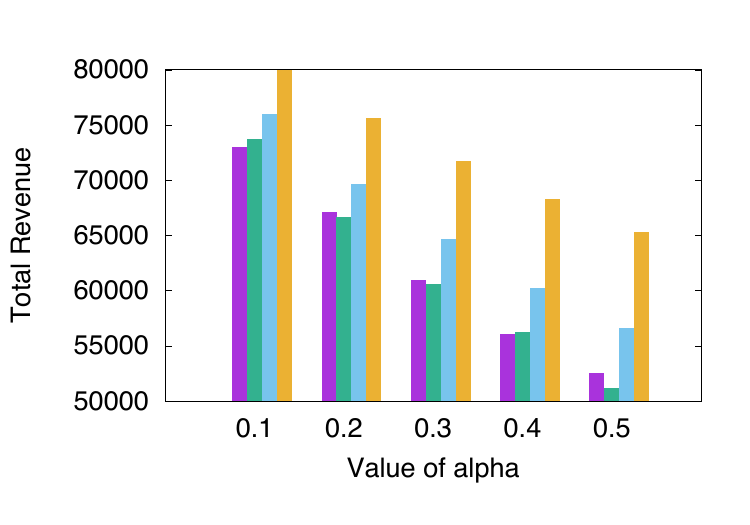}&
\hspace{-4mm}\begin{sideways} $\;$ $\;\;$ $\;\;$ $\;\;$ $\;\;$ \textsf{\small Linear} \end{sideways} \\
\vspace{-4mm}\includegraphics[width=.24\textwidth]{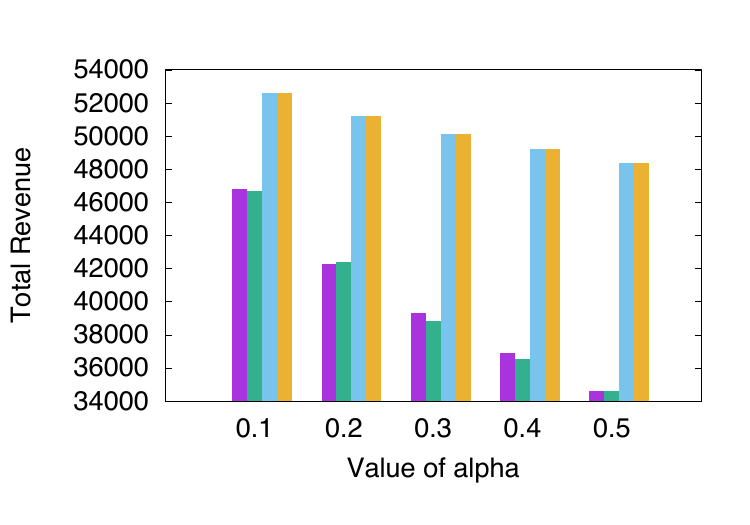}&
	 \hspace{-2mm}\includegraphics[width=.24\textwidth]{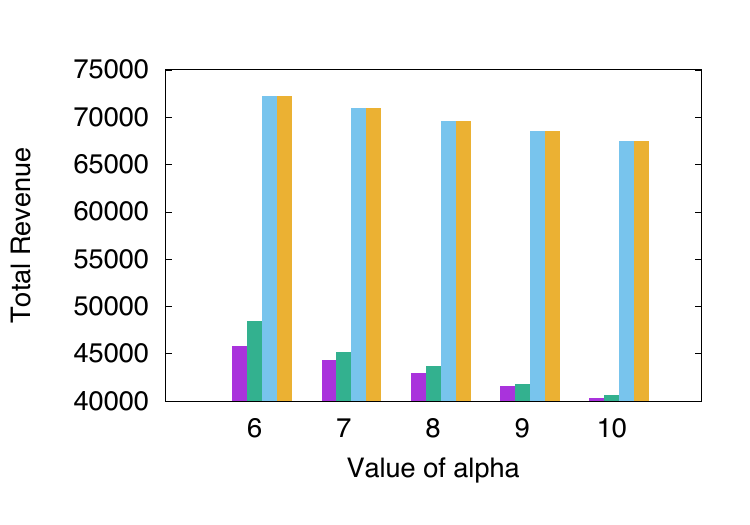} &
\hspace{-4mm}\begin{sideways}  $\;$ $\;\;$ $\;\;$ $\;\;$ $\;\;$ \textsf{\small Constant} \end{sideways}\\
\vspace{-4mm}\includegraphics[width=.24\textwidth]{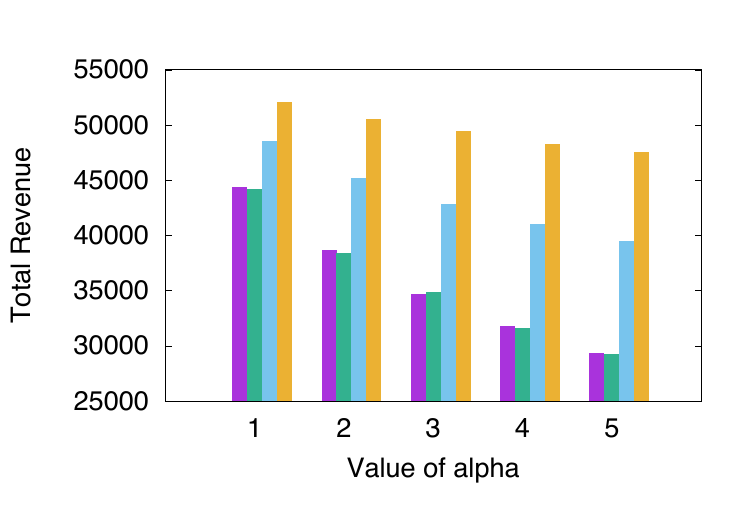}&
	 \hspace{-2mm}\includegraphics[width=.24\textwidth]{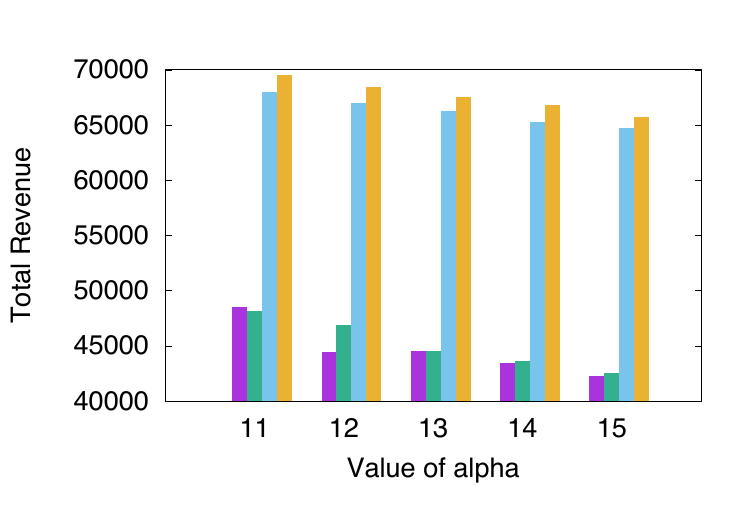} &
	  \hspace{-4mm}\begin{sideways} $\;$ $\;\;$ $\;\;$ $\;\;$ $\;\;$ \textsf{\small Sublinear} \end{sideways} \\
\includegraphics[width=.24\textwidth]{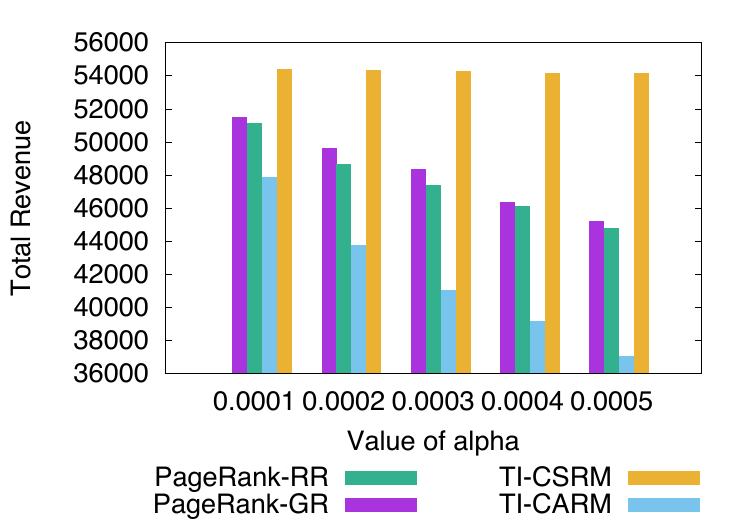}&
	 \hspace{-2mm}\includegraphics[width=.24\textwidth]{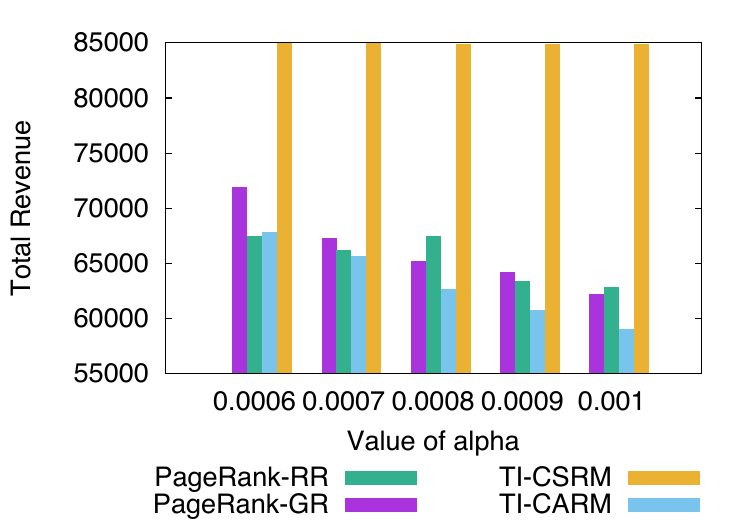} &
	  \hspace{-4mm}\begin{sideways} $\;$ $\;\;$ $\;\;$ $\;\;$ $\;\;$ \textsf{\small Superlinear} \end{sideways} \\ 	
\flix & \epi &
\end{tabular}
\caption{Total revenue as a function of $\alpha$, on \flix\ (left) and \epi\ (right), for linear, constant, sublinear, and superlinear incentive models.}
\label{fig:revenueAlphas}
\end{figure}
\eat{
\begin{figure}[t!]
\vspace{-4mm}
\begin{tabular}{ccc}
\vspace{-4mm}\includegraphics[width=.24\textwidth]{plots/alpha-nolegend/flix_alpha_totalRev_linear}&
	\hspace{-2mm}\includegraphics[width=.24\textwidth]{plots/alpha-nolegend/epi_alpha_totalRev_linear}&
\hspace{-4mm}\begin{sideways} $\;$ $\;\;$ $\;\;$ $\;\;$ $\;\;$ \textsf{\small Linear} \end{sideways} \\
\vspace{-4mm}\includegraphics[width=.24\textwidth]{plots/alpha-nolegend/flix_alpha_totalRev_uniform}&
	 \hspace{-2mm}\includegraphics[width=.24\textwidth]{plots/alpha-nolegend/epi_alpha_totalRev_uniform} &
\hspace{-4mm}\begin{sideways}  $\;$ $\;\;$ $\;\;$ $\;\;$ $\;\;$ \textsf{\small Constant} \end{sideways}\\
\includegraphics[width=.24\textwidth]{plots/alpha/flix_alpha_totalRev_sublinear}&
	 \hspace{-2mm}\includegraphics[width=.24\textwidth]{plots/alpha/epi_alpha_totalRev_sublinear} & \hspace{-4mm}\begin{sideways} $\;$ $\;\;$ $\;\;$ $\;\;$ $\;\;$ \textsf{\small Sublinear} \end{sideways} \\
\flix & \epi &
\end{tabular}
\caption{Total revenue as a function of $\alpha$, on \flix\ (left) and \epi\ (right), for linear (top), constant (middle), and sublinear (bottom) incentive models.}
\label{fig:revenueAlphas}
\end{figure}
}

On \flix and \epi we used Monte Carlo simulations (5K runs\footnote{\scriptsize We didn't observe any significant change in the influence spread estimation beyond 5K runs  for both datasets.}) to compute $\sigma_i(\{u\})$.  On \dblp and \livej, we use the out-degree of the nodes as a proxy to $\sigma_i(\{u\})$ due to the prohibitive computational cost of Monte Carlo simulations.


\smallskip\noindent\textbf{Algorithms.}
We compared \LL{four} algorithms in total. Wherever applicable, we set the parameter $\varepsilon$ to be $0.1$ for quality experiments on \flix and \epi, and $0.3$ for scalability experiments on \dblp and \livej, following the settings used in \cite{tang14}.

\squishlisttight
\item \fastcs (Algorithm~\ref{alg:fastCS}) that uses Algorithm~\ref{alg:rrBestCSNode} to find the best (cost-sensitive) candidate node for each advertiser (line~\ref{line:greedySelectBest}), and selects among those the (node, advertiser) pair that provides the maximum rate of marginal gain in revenue per marginal gain in advertiser's payment (line~\ref{line:greedyCriterCS}).
\item \fastca: Cost-agnostic version of Algorithm~\ref{alg:fastCS} that uses Algorithm~\ref{alg:rrBestCANode} to find the best (cost-agnostic) candidate node for each advertiser (replacing line~\ref{line:greedySelectBest}), and selects among those the (node, advertiser) pair with the maximum increase in the revenue of the host (replacing line~\ref{line:greedyCriterCS}).
\item \LL{PageRank-GR}: A baseline that selects a candidate node for each advertiser based on the ad-specific PageRank ordering of the nodes (replacing line~\ref{line:greedySelectBest}), and selects among those the (node, advertiser) pair that provides the maximum increase in the revenue of the host (replacing line~\ref{line:greedyCriterCS}). \textcolor{black}{Since the selection is made greedily, we refer to this algorithm as PageRank-GR.}
\item \LL{PageRank-RR: Another PageRank-based baseline that selects a candidate node for each advertiser based on the ad-specific PageRank ordering of the nodes (replacing line~\ref{line:greedySelectBest}), and uses a \textcolor{black}{Round-Robin (RR in short)} ordering of the advertisers for the assignment of their candidates into their seed sets.}
\squishend

\smallskip\noindent\textbf{Revenue vs.\ $\alpha$.}
\LL{We first compare the total revenue achieved by the four algorithms for four different seed incentive models and with varying levels of $\alpha$ (Figure~\ref{fig:revenueAlphas}). Recall that by definition, a smaller 
$\alpha$ value indicates lower 
seed costs for all users. Across all different values of $\alpha$ and all seed incentive models, it can be seen that \fastcs consistently achieves the highest revenue, often by a large margin, which increases as $\alpha$ grows. For instance, on \epi, when $\alpha=0.5$, \fastcs achieved 15.3\%, {24.3\%}, 27.6\% more revenue than \fastca, {PageRank-RR, and PageRank-GR} respectively on the linear incentive model, while these values for superlinear incentive model respectively are 25.2\%, {25.8\%}, 18.1\%. Notice that for the constant incentive model, the advantage of being cost-sensitive is nullified, hence \fastca and \fastcs end up performing identically as expected.}
 \LL{Figure~\ref{fig:costAlphas} reports the cost-effectiveness of the algorithms. Across all different values of $\alpha$ and all incentive models, it can be seen that \fastcs consistently achieves the lowest total seed costs. This is as expected, since its seed allocation strategy takes into account revenue obtained per seed user cost.}

\AR{Notice that in three of the test cases, i.e., linear seed incentives on \flix and superlinear seed incentives on both datasets, \fastca has slightly worse performance than the two PageRank-based heuristics (e.g., about 4--7\% drop in revenue).} This can be explained by the fact that, while \fastca picks seeds of high spreading potential (i.e., highest marginal revenue) without considering costs, the two PageRank-based heuristics may instead select seeds of low quality (i.e., low marginal revenue), but also of very low cost. This might create a situation in which the PageRank-based heuristics may select many more seeds, but with a smaller total seed cost than \fastca, hence, allowing the budget to be spent more on engagements that translate to higher revenue, mimicking the cost-sensitive behavior. On the other hand \fastcs always spends the given budget judiciously by selecting seeds with the best rate of marginal revenue per cost. Thus, it is able to use the budget more intelligently, which explains its superiority in all test cases. This hypothesis is confirmed by our experiments. E.g., on \flix with linear seed incentives, we observed that the average values of marginal gain in revenue, seed user cost, and rate of marginal gain per cost obtained by PageRank-GR were respectively $2.67$, $0.44$, and $7.48$, while the corresponding numbers  for \fastca were $13.47$, $2.7$, and $4.89$, and those for \fastcs were $1.28$, $0.12$, and $9.95$ respectively. While the two PageRank-based heuristics could obtain higher revenue than \fastca on \flix with linear and superlinear incentives, and on \epi with superlinear incentives, they were greatly outperformed by \fastca, hence \fastcs, in the other incentive models, showing that such heuristics are not robust to different seed incentive models, and can only get ``lucky" to the extent they can mimic the cost-sensitive behavior.

\begin{figure}[t!]
\vspace{-4mm}
\begin{tabular}{ccc}
\vspace{-4mm}\includegraphics[width=.24\textwidth]{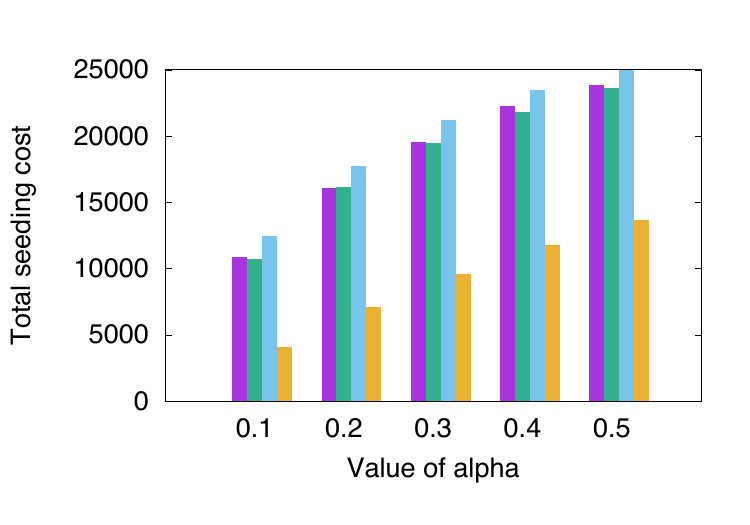}&
	\hspace{-2mm}\includegraphics[width=.24\textwidth]{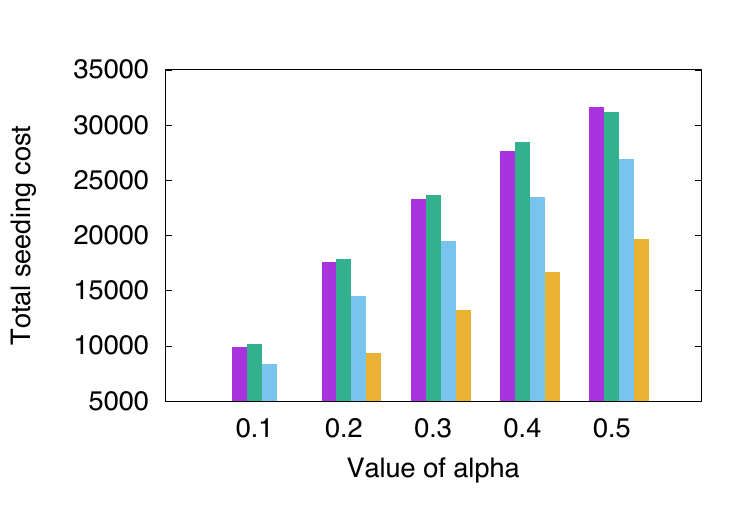}&
\hspace{-4mm}\begin{sideways} $\;$ $\;\;$ $\;\;$ $\;\;$ $\;\;$ \textsf{\small Linear} \end{sideways} \\
\vspace{-4mm}\includegraphics[width=.24\textwidth]{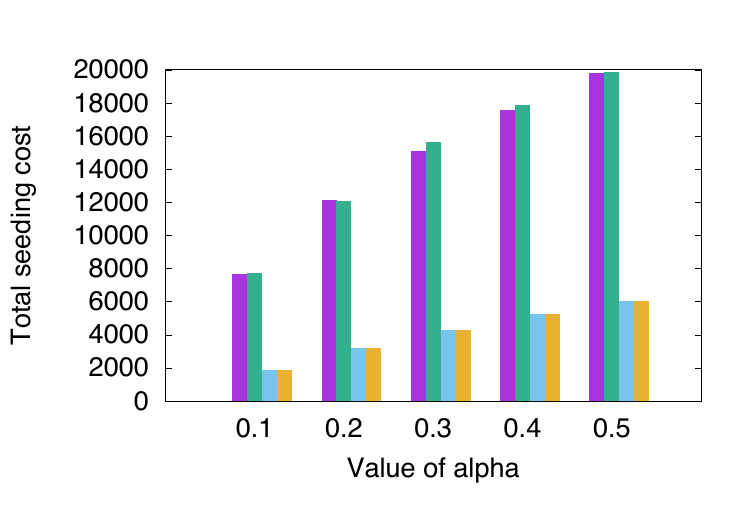}&
	 \hspace{-2mm}\includegraphics[width=.24\textwidth]{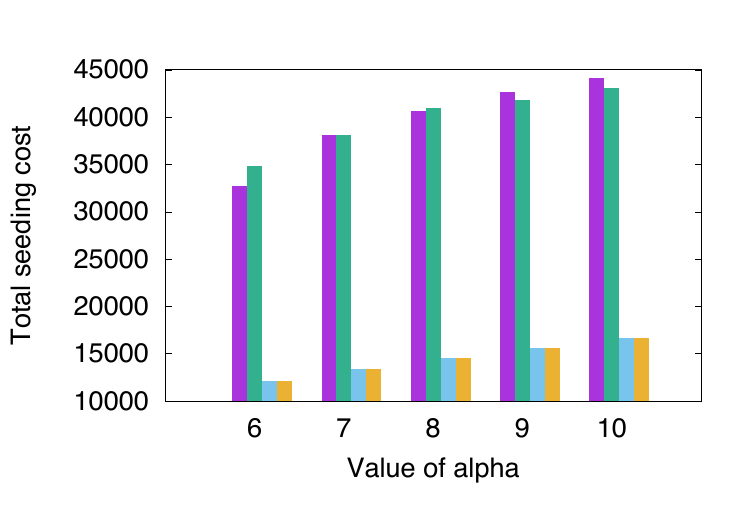} &
\hspace{-4mm}\begin{sideways}  $\;$ $\;\;$ $\;\;$ $\;\;$ $\;\;$ \textsf{\small Constant} \end{sideways}\\
\vspace{-4mm}\includegraphics[width=.24\textwidth]{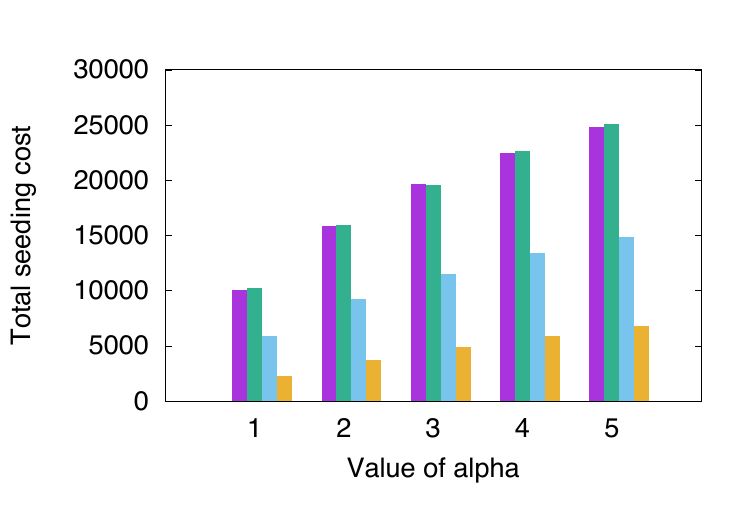}&
	 \hspace{-2mm}\includegraphics[width=.24\textwidth]{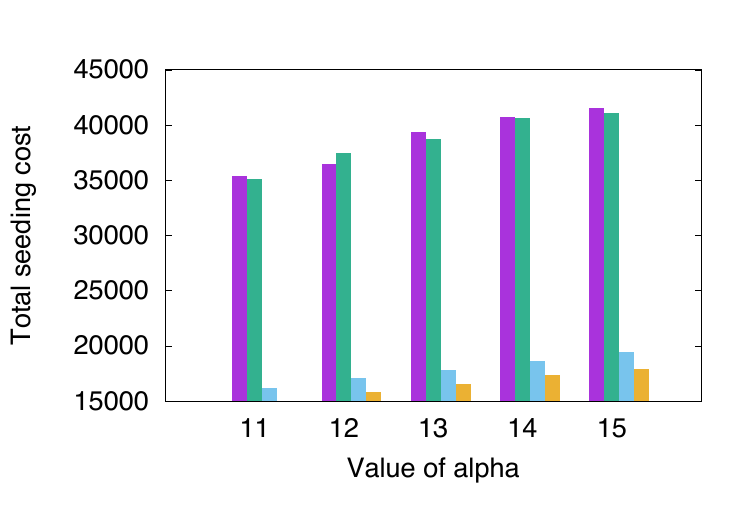} & \hspace{-4mm}\begin{sideways} $\;$ $\;\;$ $\;\;$ $\;\;$ $\;\;$ \textsf{\small Sublinear} \end{sideways} \\
\includegraphics[width=.24\textwidth]{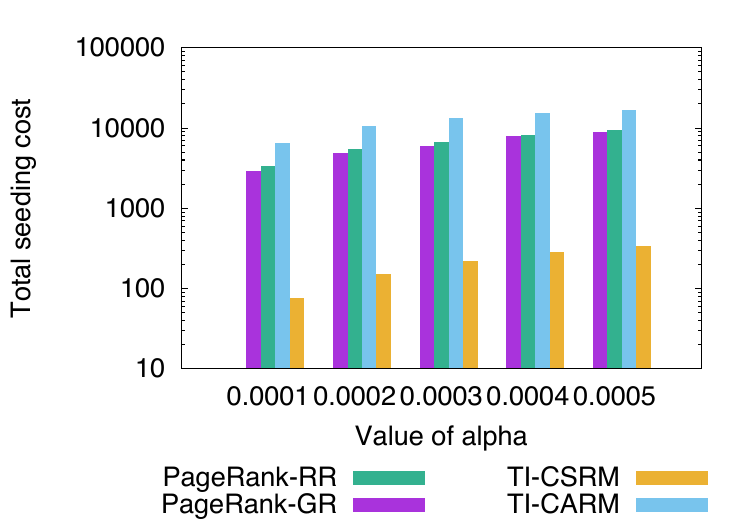}&
	 \hspace{-2mm}\includegraphics[width=.24\textwidth]{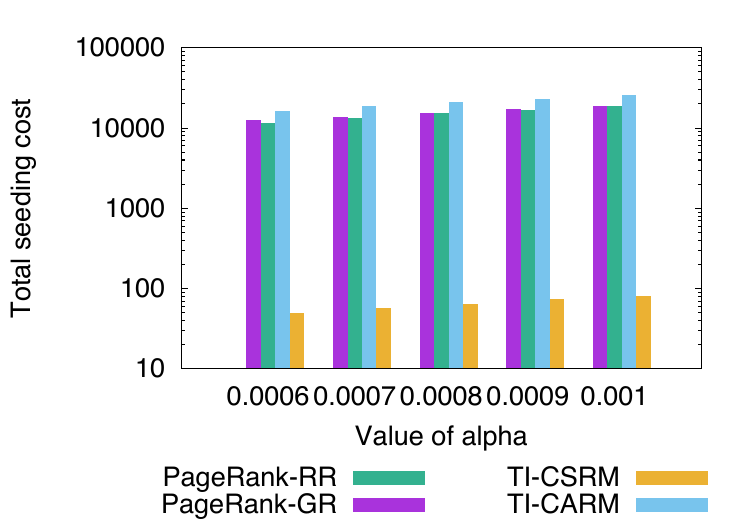} & \hspace{-4mm}\begin{sideways} $\;$ $\;\;$ $\;\;$ $\;\;$ $\;\;$ \textsf{\small Superlinear} \end{sideways} \\
\flix & \epi &
\end{tabular}
\caption{Total seeding cost as a function of $\alpha$, on \flix\ (left) and \epi\ (right), for linear, constant, sublinear, and superlinear cost models.}
\label{fig:costAlphas}
\end{figure}
\eat{
\begin{figure}[t!]
\vspace{-4mm}
\begin{tabular}{ccc}
\vspace{-4mm}\includegraphics[width=.24\textwidth]{plots/alpha-nolegend/flix_alpha_totalCost_linear}&
	\hspace{-2mm}\includegraphics[width=.24\textwidth]{plots/alpha-nolegend/epi_alpha_totalCost_linear}&
\hspace{-4mm}\begin{sideways} $\;$ $\;\;$ $\;\;$ $\;\;$ $\;\;$ \textsf{\small Linear} \end{sideways} \\
\vspace{-4mm}\includegraphics[width=.24\textwidth]{plots/alpha-nolegend/flix_alpha_totalCost_uniform}&
	 \hspace{-2mm}\includegraphics[width=.24\textwidth]{plots/alpha-nolegend/epi_alpha_totalCost_uniform} &
\hspace{-4mm}\begin{sideways}  $\;$ $\;\;$ $\;\;$ $\;\;$ $\;\;$ \textsf{\small Constant} \end{sideways}\\
\includegraphics[width=.24\textwidth]{plots/alpha/flix_alpha_totalCost_sublinear}&
	 \hspace{-2mm}\includegraphics[width=.24\textwidth]{plots/alpha/epi_alpha_totalCost_sublinear} & \hspace{-4mm}\begin{sideways} $\;$ $\;\;$ $\;\;$ $\;\;$ $\;\;$ \textsf{\small Sublinear} \end{sideways} \\
\flix & \epi &
\end{tabular}
\caption{Total seeding cost as a function of $\alpha$, on \flix\ (left) and \epi\ (right), for linear (top), constant (middle), and sublinear (bottom) cost models.}
\label{fig:costAlphas}
\end{figure}
}

\CA{Finally, as shown in Figure~\ref{fig:revenueAlphas}, the extent to which  \fastcs outperforms \fastca on both datasets is higher with linear incentives than with sublinear incentives. For instance, on \flix, \fastcs achieved $45\%$ more revenue than \fastca in the linear model, while this improvement drops to $20\%$ in the sublinear model. To understand how the seeds' expensiveness levels affect this improvement, we checked the values of singleton payments and found that the maximum singleton payment ($\rho_{max}$) is $1347$ times more expensive than the minimum singleton payment ($\rho_{min}$) in the linear model, while it is $725$ times more expensive in the sublinear model that has lower improvement rate. This relation is expected as higher variety in the expensiveness levels of the seeds require to use the budget more cleverly, hence, with more cost-effective strategies. Notice that this finding is also in line with our discussion following the proof of Theorem~\ref{theo:cs-earm1}.}

It is also worth noting that, from Figure~\ref{fig:costAlphas}, \fastcs is two to three orders of magnitude more cost-efficient than the rest in the superlinear model, and this gap is larger than that attained in linear, constant, and sublinear scenarios.



\smallskip\noindent\textbf{Revenue \& running time vs.\ window size.}
\textcolor{black}{Hereafter all presented results will be w.r.t.\ linear seed incentives, unless otherwise noted.}
As stated before in Section~\ref{sec:algorithms}, \fastcs needs to compute $\sigma_i(v | S^{t-1}_i)$, $\forall v : (v,i) \in \mathcal{E}^{t-1}$ while $u_i^{t}$ might even correspond to the node that has the \emph{minimum} marginal gain in influence spread for iteration $t$.
To have a closer look at how the revenue evolves when the seed selection criterion changes from cost-agnostic to cost-sensitive, we restrict  \fastcs to find the best cost-sensitive candidate nodes for each advertiser (line~\ref{line:greedySelectBest}) among only the $w$ nodes that have the highest marginal gain in revenue at each iteration. We refer to $w$ as the ``window size''.
Notice that \fastca corresponds to the case when $w=1$, i.e., in this case,  \fastcs inspects only the node with the maximum marginal gain in revenue.

We report the results of \fastcs with various window sizes in Fig.~\ref{fig:windowTradeoff}, which depicts the revenue vs. running time tradeoff.
Each figure corresponds to one dataset and one particular $\alpha$ value.
The $X$-axis is in log-scale.
As expected, the maximum revenue is achieved when \fastcs implements the full window $w = n$, \emph{i.e.}, when all the (feasible) nodes are inspected at each iteration for each advertiser.
The running time can go up quickly as the window size increases to $n$.
This is expected as the seed nodes selected do not necessarily provide high marginal gain in revenue, thus, \fastcs needs to use higher number of seed nodes, hence, much more RR-sets to achieve accuracy, compared to \fastca.

\begin{figure*}[t!]
\vspace{-4mm} 
\begin{tabular}{cccc}
    \includegraphics[width=.24\textwidth]{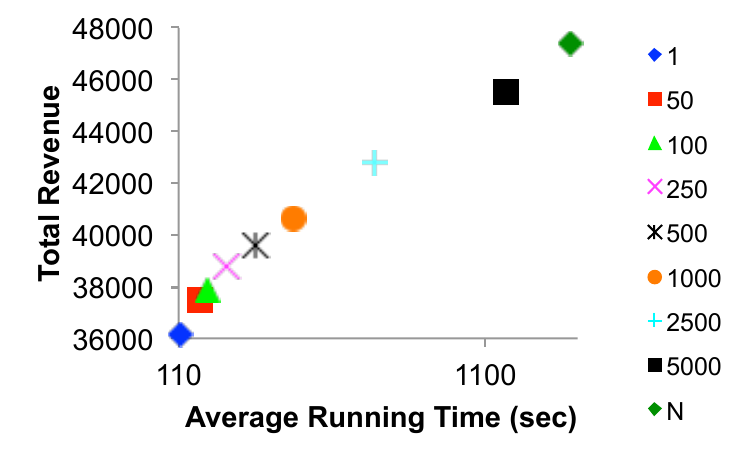}&
    \hspace{-2mm}\includegraphics[width=.24\textwidth]{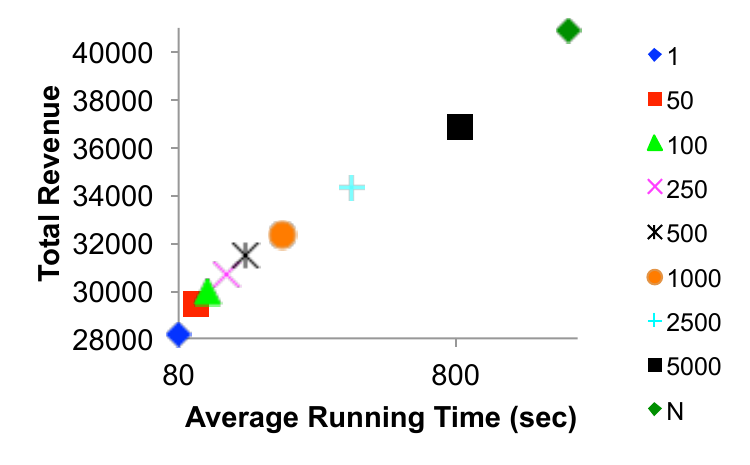}&
    \includegraphics[width=.24\textwidth]{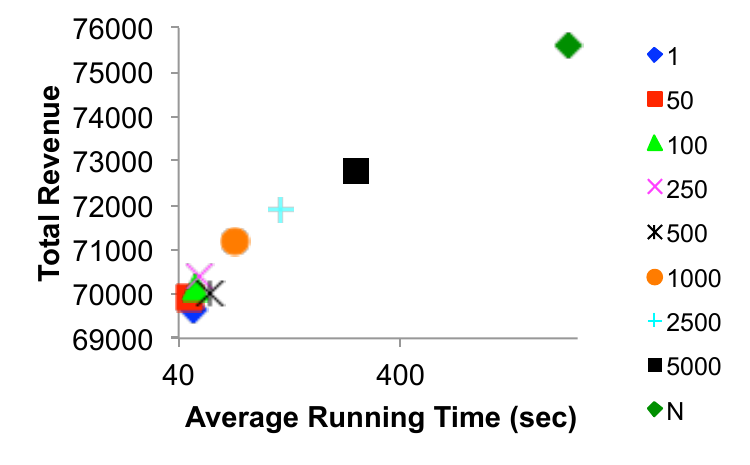}&
    \hspace{-2mm}\includegraphics[width=.24\textwidth]{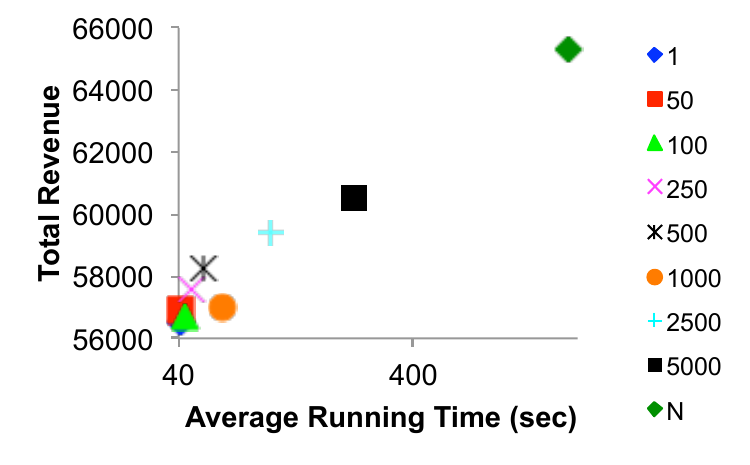}\\
	(a) \flix ($\alpha=0.2$)  & (b) \flix ($\alpha=0.5$) & (c)  \epi ($\alpha=0.2$) & (d) \epi ($\alpha=0.5$)   \\
\end{tabular}
\caption{Revenue vs running time tradeoff on \flix and \epi for two different value of $\alpha$.}
\label{fig:windowTradeoff}
\end{figure*}

\smallskip\noindent\textbf{Scalability.} We tested the scalability of \fastca and \fastcs on two larger graphs, \dblp and \livej. In all scalability experiments, we use a window size of $w = 5000$ nodes for \fastcs due to its good revenue vs running time trade-off. For simplicity, all CPEs were set to $1$.
The influence probability on each edge $(u,v)\in E$ was computed using the Weighted-Cascade model~\cite{kempe03}, where $p_{u,v}^i = 1/{|N^{in}(v)|}$ for all ads $i$.
We set $\alpha=0.2$ and $\varepsilon = 0.3$.
This setting is well-suited for testing scalability as it simulates a fully competitive case: all advertisers compete for the same set of influential users (due to all ads having the same distribution over the topics), and hence it will ``stress-test'' the algorithms by prolonging the seed selection process.

Figure~\ref{fig:time}(a) and \ref{fig:time}(b) depict the running time of \fastca and \fastcs as the number of advertisers goes up from 1 to 20, while the budget is fixed (10K for \dblp and 100K for \livej).
As can be seen, the running time increases mostly in a linear manner, and \fastcs is only slightly slower than \fastca.
Figure~\ref{fig:time}(c) and \ref{fig:time}(d) depict the running time of \fastca and \fastcs as the budget increases, while the number of advertisers is fixed at $h=5$ .
We can also see that the increasing trend is mostly linear for \fastcs, while \fastca's time goes in a flatter fashion.
All in all, both algorithms exhibit decent scalability.

Table~\ref{table:mem} shows the memory usage of \fastca and \fastcs when $h$ increases. \fastcs in general needs to use higher memory than \fastca due to its requirement to generate more RR sets that ensures accuracy for using higher seed set size than \fastca.
On \dblp, \fastca and \fastcs respectively uses a total of $4676$ and $7276$ seed nodes for $h=20$.
On \livej \fastcs used typically between 20\% to 40\% more memory than \fastca: \fastca and \fastcs respectively uses a total of $4327$ and $6123$ seed nodes for $h=20$.


\begin{table}
 \vspace{-4mm}
\centering
\small
 \caption{Memory usage (GB).\label{table:mem}}
 
  \begin{tabular}{|c|c|c|c|c|c|}
    \hline
    {\bf \dblp} & $h = 1$ & $5$ & $10$ & $15$ & $20$ \\ \hline
    \fastca &  1.6 & 7.5  & 14.9 & 22.4 &  29.8 \\ \hline
    \fastcs (5000) & 1.6 & 7.6 & 15.1 & 22.7 & 30.2 \\ \hline
    \hline
    {\bf \livej} & $h = 1$ & $5$ & $10$ & $15$ & $20$ \\ \hline
    \fastca & 2.5 & 12.1 & 25.3 & 39.4 &  54.4 \\ \hline
    \fastcs (5000) & 3.4 & 15.9 & 31.2 & 49.1 & 67.5  \\ \hline
  \end{tabular}
 \vspace{2mm}
\end{table}

\begin{figure*}
\begin{tabular}{cccc}
    \includegraphics[width=.23\textwidth]{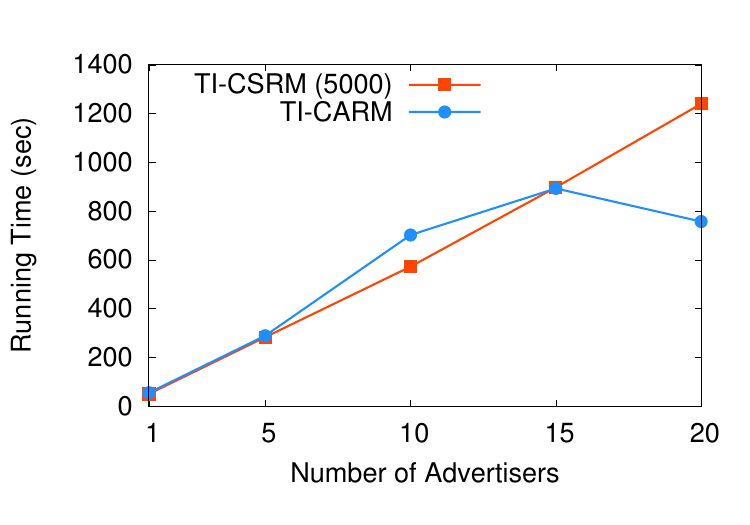}&
    \hspace{2mm}\includegraphics[width=.23\textwidth]{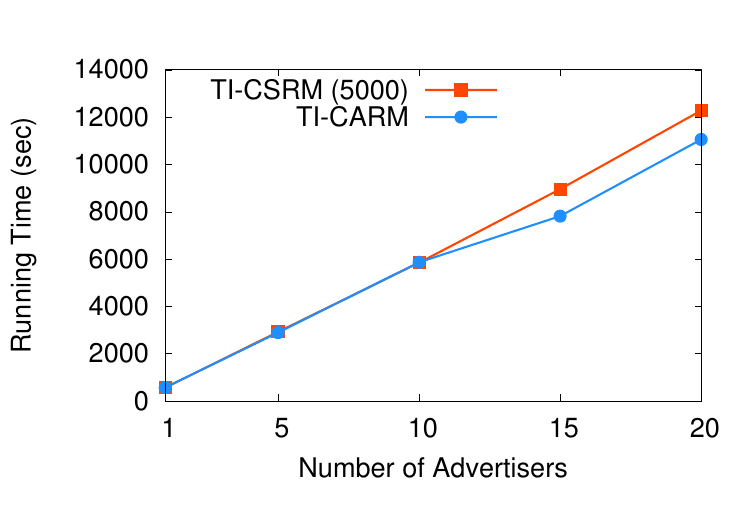}&
    \hspace{2mm}\includegraphics[width=.23\textwidth]{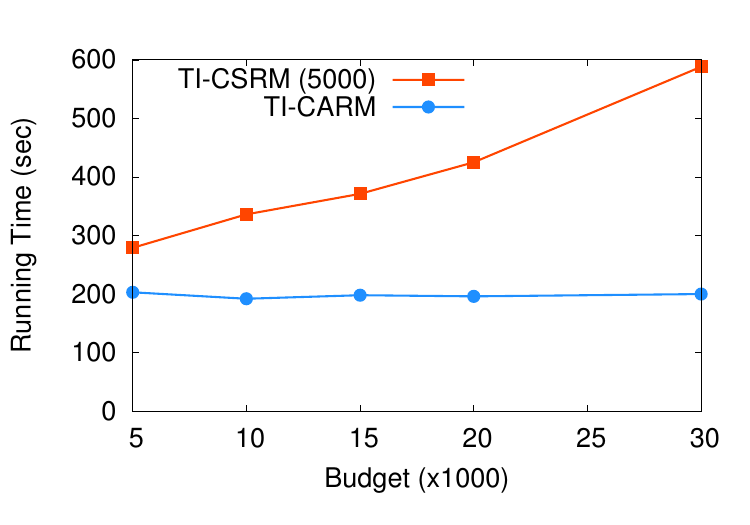}&
        \hspace{2mm}\includegraphics[width=.23\textwidth]{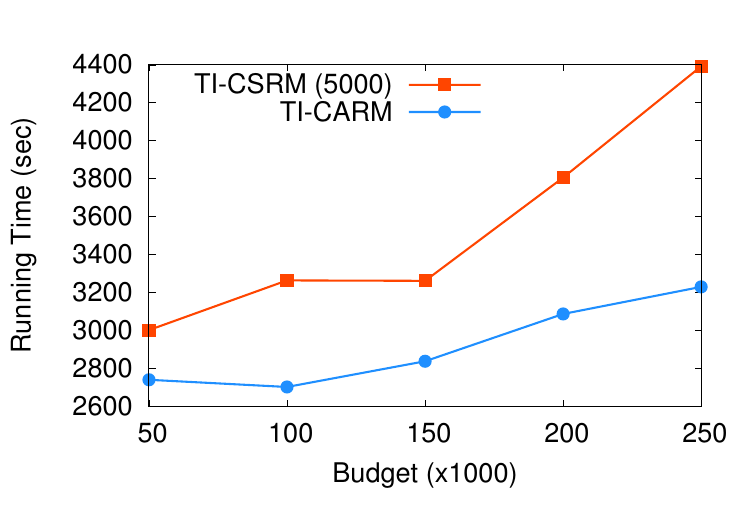}\\
	(a) \dblp ($h$) & (b) \livej ($h$)   & (c) \dblp (budgets)  & (d) \livej (budgets)  \\
\end{tabular}
\caption{Running time of \fastca and \fastcs on \dblp and \livej}
\label{fig:time}
\end{figure*}

\section{Related work}
\label{sec:related}
\enlargethispage{\baselineskip}

\spara{Computational advertising.}
\eat{
As advertising on the web has become one of the largest businesses during the last decade, the general area of \emph{computational advertising} has attracted a lot of research interest. The central problem of computational advertising is to find the ``best match'' between a given user in a given context and a suitable advertisement. The context could be a user entering a query in a search engine (``sponsored search"), reading a web page (``content match" and ``display ads"), or watching a movie on a portable device, etc.}
Considerable work has been done in sponsored search and display ads~\cite{goel08, feldman09,feldman10,mirrokni12,devanur12,mehta13}.
In sponsored search,
\eat{search engines show ads deemed relevant to user queries,
in the hope of maximizing click-through rates and in turn, revenue.}
revenue maximization is formalized as
the well-known {\em Adwords} problem~\cite{adwords}.
Given a set of keywords and bidders with their daily budgets and bids for each keyword, words need to be assigned to bidders upon arrival, to maximize the revenue for the day, while respecting bidder budgets.
This can be solved with a competitive ratio of $(1-1/e)$~\cite{adwords}.


\spara{Social advertising.} In comparison with computational advertising, social advertising is in its infancy. Recent efforts, including Tucker \cite{tucker12} and Bakshy et al.~\cite{bakshy12}, have shown, by means of field studies on sponsored posts in Facebook's News Feed,
the importance of taking social influence into account when developing social advertising strategies.
However, literature on exploiting social influence for social advertising is rather limited.
Bao and Chang have proposed \emph{AdHeat}~\cite{AdHeat}, a social ad model considering social influence
 in addition to relevance for matching ads to users.
Their experiments 
show that AdHeat significantly outperforms the relevance model on click-through-rate (CTR). Wang et al.~\cite{wang2011learning} propose a new model for learning relevance and apply it for selecting relevant ads for Facebook users.
\eat{
In both~\cite{AdHeat}~and~\cite{wang2011learning},  the proposed model is just assessed as a relevance model in terms of CTR: }
Neither of these works studies viral ad propagation or revenue maximization.

\eat{
The three papers most related to our proposal have been published in 2015. We next discuss them and how our proposal differentiates and advances beyond this prior work.
}

Chalermsook et al.~\cite{chalermsook} study revenue maximization for the host, when dealing with multiple advertisers. In their setting, each advertiser 
pays the host an amount 
for each product adoption, up to a 
budget.
In addition, each advertiser also specifies the maximum size 
of its seed set. This additional constraint considerably simplifies the problem compared to our setting, where the absence of a prespecified seed set size is a \emph{significant challenge}.
\eat{
Thus, in practice, they have a double budget: one on the size of the seed set, one on the total CPE. Not having seed set size specified beforehand is a {\em significant challenge} we address in our work.
}
\eat{Another key difference with our work is that their propagation model is not topic-aware. A topic-aware propagation model cleanly distingishes between ads that are orthogonal and those that compete for seed users, depending on their proximity in the topic space.}
\eat{
Another important difference is that in our model both ads and social influence are \emph{topic-aware}: this produces an interesting natural competition among ads which are close in a topic space for the attention of the users which are influential in the same area of the topic space. Instead Chalermsook et al.~\cite{chalermsook} adopt the simple IC model where all the ads are exactly the same.
}

Aslay et al.~\cite{AslayLB0L15} study regret minimization for a host supporting campaigns from multiple advertisers. Here, regret is the difference between the monetary budget of an advertiser and the value of expected number of engagements achieved by the campaign, based on the CPE pricing model. They share with us the pricing model and advertiser budget. However, they do not consider seed user costs. Besides they attack a very different optimization problem and their algorithms and results do not carry over to our setting.

\eat{
also study social advertising through the viral-marketing lenses and show that keeping into account the propensity of ads for viral propagation can achieve significantly better performance. However, uncontrolled virality could be undesirable for the host as it creates room for exploitation by the advertisers: hoping to tap uncontrolled virality, an
advertiser might declare a lower budget for its marketing campaign,
aiming at the same large outcome with a smaller cost. Therefore Aslay et al. study the problem of \emph{regret minimization} for ads allocation, where regret is the absolute value of the difference between the budget of one advertiser and the total cost paid by the advertiser to the host based on CPE pricing model.
Therefore they study a different optimization problem and their business model is different (as they do not consider an explicit monetary incentive to the seed users).
}

Abbassi et al.~\cite{AbbassiBM15} 
study a cost-per-mille (CPM) model in display advertising. The host enters into a contract with each advertiser to show their ad to a fixed number of users, for an agreed upon CPM amount per thousand impressions. The problem is that of selecting the sequence of users to show the ads to, in order to maximize the expected number of clicks. This is a substantially different problem which they show is APX-hard and propose heuristic solutions.

\eat{
differently from our work and~\cite{AslayLB0L15,chalermsook} which are all absed on CPE, they consider a CPM (cost-per-mille) pricing model: i.e., the advertiser enters in a contract with the host for its ad to be shown to a fixed
number of users, agreeing to pay a certain CPM amount for every thousand impressions. Under this model the number of engagements (or clicks) that the ad receives, does not directly influences the revenue of the host. However, optimizing click-through-rate is nevertheless an important goal as it makes more likely that the advertiser will come back for another advertising campaign. Therefore the problem studied by Abbassi et al~\cite{AbbassiBM15} is that of allocating ads to users so to optimize the number of clicks, for a predefined number of ads impressions, keeping in consideration social influence. Their results are mostly of theoretical interest and negative nature (i.e., hardness and strong inapproximability).
}

Alon et al.~\cite{alon-etal-opt-budget-www2012} study budget allocation among channels and influential customers, with the intuition that a channel assigned a higher budget will make more attempts at influencing customers. They do not take into account viral propagation. 
Their main result is that for some influence models the budget allocation problem can be approximated, while for others it is inapproximable. 
Notably, none of these previous works studies \emph{incentivized social advertising} where the seed users are paid monetary incentives.


\spara{Viral marketing.}
Kempe et al.~\cite{kempe03} formalize the influence maximization problem which requires to select $k$ seed nodes, where $k$ is a cardinality budget, such that the expected spread of influence from the selected seeds is maximized. Of particular note are the recent advances (already reviewed in Section~\ref{sec:algorithms}) that have been made in designing scalable approximation algorithms~\cite{borgs14,tang14,cohen14,tang2015influence,NguyenTD16} for this hard problem. Numerous variants of the influence maximization problem have been studied over the years, including competition~\cite{BharathiKS07,CarnesNWZ07}, host perspective~\cite{lu2013bang,AslayLB0L15}, non-uniform cost model for seed users~\cite{LeskovecKDD07,nguyen2013budgeted}, and fractional seed selection~\cite{DemaineHMMRSZ14}. However, to our knowledge, there has been no previous work that addresses incentivized social advertising, while leveraging viral propagation of social ads and handling advertiser budgets.

\eat{
The key difference between this literature and our setting, is that in the standard influence maximization the budget of an advertiser is modeled as a cardinality constraint on the number of free products to offer, hence the number of seed users to target~\cite{kempe03}. Some work has studied the possibility to target the seed users non-uniformly, that is to say that different seeds might have a different cost, in which case the budget of an advertiser is modeled as a monetary amount that will be spend on the non-uniform costs of incentivizing seed users\cite{LeskovecKDD07}. On the other hand, the real-world social advertisement models operate with monetary budgets, which is used not only for incentivizing the seed users, but more generally for paying the CPE. Hence, the classic treatment of budgets in the optimization of a viral marketing campaign is inadequate for modeling the real-world social advertising scenarios.
}

\section{Conclusions}
\label{sec:conclusions}
\enlargethispage{\baselineskip}

In this paper, we initiate the investigation of incentivized social advertising,
by formalizing the fundamental problem of revenue maximization from the
host perspective. In our formulation, incentives paid to the seed users are determined by
their demonstrated past influence in the topic of the specific ad. We show that,
keeping all important factors -- topical relevance of ads, their propensity for
social propagation, the topical influence of users, seed users' incentives, and advertiser
budgets -- in consideration, the problem of revenue maximization in incentivized
social advertising is NP-hard and it corresponds to the problem of monotone submodular
function maximization subject to a partition matroid constraint on the
ads-to-seeds allocation and multiple submodular knapsack constraints on the advertiser
budgets.
For this problem, we devise two natural greedy algorithms that differ in
their sensitivity to seed user incentive costs, provide
formal approximation guarantees, and achieve scalability by adapting to our context  recent advances made in scalable estimation of expected influence spread.

Our work takes an important first step toward enriching the framework of incentivized
social advertising with powerful ideas from viral marketing, while
making the latter more applicable to real-world online marketing. It opens up several interesting avenues for further research: \LL{$(i)$ it remains open whether our winning algorithm \fastcs can be made more memory efficient hence more scalable; $(ii)$ it remains open whether the approximation bound for \CSRM provided in Theorem~\ref{theo:cs-earm1} is tight; $(iii)$ it is interesting to integrate hard competition constraints into the influence propagation process; $(iv)$ it is worth studying our problem in an online adaptive setting where the partial results of the campaign can be taken into account while deciding the next moves.}
All  these directions  offer a
wealth of possibilities for future work.

%


\enlargethispage{\baselineskip}



\end{document}